\documentclass[12pt]{article}

\usepackage{epsfig}
\usepackage{graphicx}
\usepackage{makeidx}
\usepackage{setspace}
\usepackage{amssymb,amsmath,amsthm}
\newtheorem{theorem}{Theorem}{}
{}
\newtheorem{proposition}{Proposition}{}
{}
{}
\newtheorem{lemma}{Lemma}{}
{}
{}
{}
{}
{}
{}

\newcommand{\channel}{p_{Y|X}(y|x)}

\newcommand{\expt}{\mathbb{E}}
\newcommand{\code}{\mathcal{C}}

\newcommand{\graph}{\mathcal{G}}
\newcommand{\prob}{\mathbb{P}}

\newcommand{\openone}{\leavevmode\hbox{\small1\normalsize\kern-.33em1}}

\title{Sharp Bounds for Optimal Decoding of Low Density Parity Check Codes}

\begin{document}

%

\author{Shrinivas Kudekar and  Nicolas Macris \\
\\
Ecole Polytechnique F\'ed\'erale de Lausanne \\
School of Computer and Communication Sciences \\
LTHC, I\&C, Station 14,  CH-1015 Lausanne\\ shrinivas.kudekar@epfl.ch, nicolas.macris@epfl.ch}


\maketitle

\begin{abstract} Consider communication over a binary-input memoryless
output-symmetric channel with low density parity check (LDPC) codes and maximum
a posteriori (MAP) decoding. The replica method of spin glass theory allows to
conjecture an analytic formula for the average input-output conditional entropy
per bit in the infinite block length limit. Montanari proved a lower bound for
this entropy, in the case of LDPC ensembles with convex check degree
polynomial, which matches the replica formula. Here we extend this lower bound
to any irregular LDPC ensemble. The new feature of our work is an analysis of
the second derivative of the conditional input-output entropy with respect to
noise. A close relation arises between this second derivative  and correlation
or mutual information of codebits. This allows us to extend the realm of the
``interpolation method'', in particular we show how channel symmetry allows to
control the fluctuations of the ``overlap parameters''.  
\end{abstract}

\section{Introduction and Main Results}\label{section-1} 

Linear codes based on sparse random graphs have emerged as a major chapter of coding theory \cite{Ruediger-Ridchardson-book}. While the belief propagation (BP) decoding algorithm and density evolution method have been explored in detail because of their low algorithmic complexity and good performance, much remains to be understood about the optimal (MAP) performance bounds of sparse graph codes. Recent theoretical progress on the binary erasure channel (BEC) has convincingly shown that BP and MAP decoding have intimate relationships (see \cite{Ruediger-Ridchardson-book} and in particular \cite{Maxwell}), but understanding this relationship for other channels is still a largely open problem. In fact, the replica and/or cavity methods of statistical mechanics of dilute spin glass models allow to conjecture an analytic formula for 
$H_n(\underline{X}|\underline{Y})$, the entropy of the
transmitted message $\underline{X}=(X_1,...,X_n)$ conditional to the received message $\underline{Y}=(Y_1,...,Y_n)$ in the large block length limit $n\to +\infty$. The replica formula expresses the conditional entropy as the solution of a  variational problem whose critical points are given by the density evolution fixed point equation
(see \cite{Saad}, \cite{MonGal}). If one is to solve the fixed point equation iteratively, the choice of initial conditions is not necessarily the one given by channel outputs (as in standard density evolution) but the one which yields the maximum conditional entropy. Note that a byproduct of the replica formula is the determination of 
 the maximum a posteriori (MAP) noise threshold, above which reliable
communication is not possible whatever the decoding algorithm. 

The proof of the replica formulas is, in general, an open problem\footnote{In a few spin glass models the replica formulas have been fully demonstrated. Remarkably Talagrand \cite{Talagrand-paper} has proven the Parisi formula with full symmetry breaking \cite{Parisi-paper} for the Sherrington-Kirkpatrick (SK) model. In \cite{Korada-Macris-seattle} it is shown that the replica symmetric formula holds for a complete $p$-spin model with gauge symmetry.}. In the context of communication  they have been proven for a class of low density parity check codes (LDPC) codes on the BEC \cite{Meason-Montanari-Urbanke-nice}, \cite{Kudekar-Korada-Macris-nice} (see also \cite{Kudekar-Macris-turbo} for recent work going beyond the BEC) and for low density generator codes (LDGM) on a class of channels \cite{Kudekar-Macris-porto}.
 
A promising approach towards a general proof of the replica formulas seems to be the use of the so-called interpolation method first developped in the context of the SK model \cite{Guerra-Toninelli}, \cite{Guerra}, \cite{Talagrand-book}. 
Consider an LDPC($n,\Lambda, P$) ensemble where $\Lambda(x)=\sum_d\Lambda_d x^d$, 
$P(x)=\sum_k P_k x^k$ are the variable and check degree distributions from the node perspective. We will always assume that the maximal degrees are finite.
Montanari \cite{Montanari} (see also the related work of Franz-Leone \cite{Franz-Leone} and Talagrand- Pachenko \cite{Talagrand-Pachenko}) has developped the interpolation method for such a system and has derived a lower bound for the conditional entropy  for ensembles with any 
polynomial $\Lambda(x)$ but $P(x)$  restricted to be convex for $-e\leq x\leq e$ 
(in particular if the check degree is constant this means it has to be even). An important fact is that these lower bounds match the replica solution, and are thus believed to be tight. 
Since Fano's inequality tells us that the block error probability
for a code having length $n$ and rate $r$ is lower bounded by
$\frac{1}{rn}H_n(\underline{X}|\underline{Y})$,
 an immediate application of the lower bound is the numerical computation of a rigorous upper bound on the MAP threshold. 

In the present paper we drop the convexity requirement for $P(x)$ in the cases of the BEC, BIAWGNC with any noise level an in the case of general binary memoryless (BMS) channels in a high noise regime. In other words we prove the lower bound for any standard regular (so odd degrees are allowed) or irregular code ensemble.

Besides the main result itself, we introduce a new tool in the form of a relationship between the second derivative of the conditional entropy with respect to the noise and correlations functions of codebits. These correlation functions are shown to be intimately related to the mutual information between two codebits. The formulas are somewhat similar to those for GEXIT functions \cite{Ruediger-Ridchardson-book} which relate the first derivative of conditional entropy to soft bit estimates. By combining these relations with the interpolation method we are able to control the fluctuations of the so-called overlap parameters. This part of our analysis is crucial for proving the general lower bound on the conditional entropy and relies heavily on channel symmetry. 

A preliminary summary of the present work has appeared in \cite{Kudekar-Macris-seattle}.

\subsection{Variational  bound on the conditional entropy}
Let $\channel$ be the transition probability of a BMS$(\epsilon)$ channel where $\epsilon$ is the noise parameter (understood to vary in the appropriate range). We will work in terms of both the likelihood  
\begin{equation}\nonumber
l=\ln\biggl[\frac{p_{Y\vert X}(y\vert 0)}{p_{Y\vert X}(y\vert 1)}\biggr]
\end{equation} 
and difference 
\begin{equation}\nonumber
t= p_{Y\vert X}(y\vert 0)-p_{Y\vert X}(y\vert 1)=\tanh\frac{l}{2}
\end{equation} 
variables. It will be convenient to use the notation $c_L(l)$  and $c_D(t)$ for the distributions of $l$  and $t$, assuming that the all zero codeword is transmitted (that is to say that $c_L(l) dl = c_D(t) dt=  p_{Y\vert X}(y\vert 0)dy$).

Let $V$ be some random variable with an arbitrary density $d_V(v)$ satisfying the symmetry condition
$d_V(v)=e^v d_V(-v)$. 
Also let 
\begin{equation}\label{U-variable}
U = \tanh^{-1}\biggl[\prod_{i=1}^{k-1}\tanh V_i\biggr]
\end{equation}
where $V_i$ are i.i.d copies of $V$ and $k$ is the (random) degree of a check node. We denote by $U_c$, $c=1,...,d$ i.i.d copies of $U$ where $d$ is the (random) degree of variable nodes.
 Notice that in the belief propagation (BP) decoding algorithm $U$ appears as the
check to variable node message and $V$ appears as the variable to check node message.
Define the functional\footnote{The subscript $RS$ stands for ``replica symmetric'' because this functional has been obtained from the replica symmetric ansatz for an appropriate spin glass, see for example \cite{MonGal}, \cite{Saad}} (we view it as a functional of the probability distribution $d_V$)
\begin{align}\nonumber
h_{RS}[d_V;\Lambda, P]=&
\expt_{l,d,U_c}\biggl[\ln\biggl(
e^{\frac{l}{2}}\prod_{c=1}^{d}(1+\tanh U_c)
+e^{-\frac{l}{2}}\prod_{c=1}^{d}(1-\tanh U_c)\biggr)\biggr]
\\ \nonumber & +\frac{\Lambda'(1)}{P'(1)}\expt_{k,V_i}\biggl[\ln(1+\prod_{i=1}^{k}\tanh V_i)\biggr] 
\\ \nonumber &
-\Lambda'(1)\expt_{V,U}\biggl[\ln(1+\tanh V \tanh U)\biggr] 
-\frac{\Lambda'(1)}{P'(1)}\ln 2
\nonumber
\end{align}
Our main result is about the conditional entropy per bit, averaged over the code ensemble $\code={\rm LDPC}(n,\Lambda, P)$.
$$
\mathbb{E}_{\code}[h_n]=\frac{1}{n} \mathbb{E}_{\code}[H_n(\underline{X}|\underline{Y})]
$$ 

\vskip 0.5cm

\noindent{\bf Definition H.} {\it We define the parameters ($p$ an integer)
\begin{equation}\label{parameters}
m_0^{(2p)}=\mathbb{E}[t^{2p}],\qquad m_1^{(2p)}=\frac{d}{d\epsilon} \mathbb{E}[t^{2p}],\qquad m_2^{(2p)}=\frac{d^2}{d\epsilon^2} \mathbb{E}[t^{2p}]
\end{equation}
and say that a ${\rm BMS}(\epsilon)$ channel is in the high noise regime if the following series expansions 
\begin{equation}
 \sum_{p}  (p+1) m_0^{(2p)}
 \qquad
 \sum_{p} \bigl(\frac{5}{2}\bigr)^{2p} \vert m_1^{(2p)}\vert
 \qquad
 \sum_{p} \frac{\vert m_2^{(2p)}\vert}{2p(2p-1)}\qquad
\end{equation}
are convergent and if 
$$
(\sqrt{2} -1)\bigl(\frac{5}{2}\bigr)^2\vert m_1^{(2)}\vert > \sum_{p\geq 2} \bigl(\frac{5}{2}\bigr)^{2p} \vert m_1^{(2p)} \vert
$$
}

\vskip 0.5cm

For example the ${\rm BSC}(\epsilon)$ certainly satisfies $H$ if the crossover noise parameter is  close enough to $\frac{1}{2}$, because 
$\mathbb{E}[t^{2p}]= (1-2\epsilon)^{2p}$. More generaly {\it any channel with bounded likehood variables satisfies 
$H$ for a regime of sufficiently high noise}. For channels with unbounded likehoods the condition will be satisfied if $c_L(l)$ has sufficiently good decay properties. But note that the ${\rm BEC}(\epsilon)$ which has mass at $l=+\infty$ does not satisfy this condition since $\mathbb{E}[t^{2p}]= 1-\epsilon$. However as we will see {\it for the BEC($\epsilon$) and the ${\rm BIAWGNC}(\epsilon)$  we do not need condition $H$}. For these two channels our analysis can be made fully non-perturbative, and holds for all noise levels.

\vskip 0.5cm

\begin{theorem}[\bf{Variational Bound}]\label{thm:main}
Assume communication using a standard irregular ${\cal C}={\rm LDPC}(n,\Lambda,P)$ code ensemble, 
through a BEC($\epsilon$) or BIAWGNC($\epsilon$) with any noise level or a BMS($\epsilon$) channel satisfying $H$.
For all $\epsilon$ in the above ranges we have, 
\begin{equation}\nonumber
\liminf_{n\to +\infty}\expt_{\cal C}[h_n]\geq \sup_{d_V} h_{RS}[d_V;\Lambda, P]
\end{equation}
\end{theorem}
Let us note that this theorem already appears in \cite{MacrisIT07} for the special case of the BIAWGNC for a Poissonnian $\Lambda(x)$.
We stress again that a formal calculation using the replica method yields
\begin{equation}\nonumber
\lim_{n\to+\infty}\expt_{\cal
C}[h_n]=\sup_{d_V} h_{RS}[d_V;\Lambda,P]
\end{equation}
For this reason it is strongly suspected that the converse inequality holds as well, but so far no progress has been made except in a limited number of situations alluded to before.

\subsection{Derivatives of the conditional entropy}
 
Our proof of the variational bound uses integral formulas for the first and second derivatives of
$\expt_{\cal C}[h_n]$ with respect to the noise parameter. The ensemble formulas follow from slightly more general ones that are valid for any fixed linear code. To give the formulation for a fixed linear code it is convenient to introduce a noise vector
$\underline{\epsilon}=(\epsilon_1,...,\epsilon_n)$ and  a ${\rm BMS}(\underline{\epsilon})$ channel with noise level $\epsilon_i$ when bit $x_i$ is sent. When all noise levels are set to the same value $\epsilon$ the channel is denoted ${\rm BMS}(\epsilon)$.
The distributions of the likelihood $l_i$ or difference domain $t_i$ representations of the channel outputs now depend on $\epsilon_i$. In order to keep the notation simpler we do not explicitely indicate the $\epsilon_i$ dependence and still denote them as $c_L(l_i)$ and $c_D(t_i)$ respectively.

We introduce the soft MAP estimates of bit $X_i$ 
\begin{align*}
L_i= \ln \biggl[\frac{p_{X_i\vert \underline{Y} }(0\vert \underline{y})}{p_{X_i\vert
\underline{Y}}(1\vert \underline{y})}\biggr], \quad \quad
 T_i  =  p_{X_i\vert \underline{Y} }(0\vert \underline{y}) - p_{X_i\vert
\underline{Y}}(1\vert \underline{y})=\tanh \frac{L_i}{2} \nonumber
\end{align*}
and the soft estimate for the modulo $2$ sum $X_i\oplus X_j$,
\begin{align*}
L_{ij}= \ln \biggl[\frac{p_{X_i\oplus X_j\vert \underline{Y}}(0\vert \underline{y})}{p_{X_i\oplus X_j\vert \underline{Y}}(1\vert \underline{y})}\biggr], \quad \quad &T_{ij}= p_{X_i\oplus X_j\vert \underline{Y}}(0\vert \underline{y}) - p_{X_i\oplus X_j\vert \underline{Y}}(1\vert \underline{y}) =\tanh \frac{L_{ij}}{2} \nonumber
\end{align*}
In the sequel the notations $\underline v^{\sim i}$ (resp. $\underline v^{\sim ij}$) means that component $v_i$ (resp. $v_i$ and $v_j$) are omitted from the vector $\underline v$.
The following is known \cite{Ruediger-Ridchardson-book} but we state it for completeness. A derivation in the spirit of the present paper can also be found in \cite{Macris}.

\begin{proposition}[\bf{GEXIT formula}]\label{lem:firstderi}
For any ${\rm BMS}(\underline{\epsilon})$ channel and any fixed linear code we have
\begin{equation}\nonumber  
\frac{\partial}{\partial\epsilon_i} H_n(\underline{X}\mid \underline{Y})=
\int_{-1}^{+1}dt_i\frac{\partial c_D(t_i)}{\partial\epsilon_i} g_1(t_i)
\end{equation} 
where 
$$
g_1(t_i)=-\expt_{\underline{t}^{\sim i}}\biggl[\ln\biggl(\frac{1-t_iT_i}{1 - t_i}\biggr)\biggr]
$$
\end{proposition}

This formula will be used for an ensemble that is symmetric under permutation of bits and a ${\rm BMS}(\epsilon)$ channel. Using 
\begin{equation}\nonumber
 \frac{d}{d\epsilon} H_n(\underline X\mid \underline Y)= \sum _{i=1}^n \frac{\partial}{\partial\epsilon_i} H_n(\underline X\mid \underline Y)\bigg\vert_{\epsilon_i=\epsilon}
\end{equation}
and averaging over the code ensemble ${\cal C}$ we get for the average entropy per bit,
\begin{equation}\nonumber
\frac{d}{d\epsilon} \mathbb{E}_{\code} [h_n]=
\int_{-1}^{+1}dt_1\frac{\partial c_D(t_1)}{\partial\epsilon} \mathbb{E}_{\cal C}[g_1(t_1)]
\end{equation} 
There are two channels where these general formulas take a simpler form.
For the BEC\footnote{In this case the ratio in the logarithm may take the ambiguous value $\frac{0}{0}$ but the formula is to be interpreted as \eqref{eq:becfderi}. We will see in section \ref{section-2} that in terms of extrinsic soft bit estimates there is an analogous expresion that is unambiguous.}
\begin{equation}\label{eq:becfderi}
\frac{\partial}{\partial\ln\epsilon_i}  H_n(\underline{X}\mid \underline{Y})= 
\ln 2(1-\mathbb{E}_{\underline{t}}[T_i])
\end{equation}
and 
\begin{equation}\label{bec-intrin}
\frac{d}{d\ln\epsilon} \mathbb{E}_{\code} [h_n]= 
\ln 2(1-\mathbb{E}_{\code,\underline{t}}[T_1])
\end{equation}
Similarly on the BIAWGNC,
\begin{equation}\label{eq:awgnfderi}
 \frac{\partial}{\partial\epsilon_i^{-2}} H_n(\underline{X}\mid \underline{Y})= -\frac{1}{2}
(1-\mathbb{E}_{\code,\underline{l}}[T_i])
\end{equation}
and
\begin{equation}\label{biawgnc-intrin}
\frac{d}{d\epsilon^{-2}} \mathbb{E}_{\code} [h_n]= -\frac{1}{2}
(1-\mathbb{E}_{\code,\underline{t}}[T_1]) 
\end{equation}
We will prove
\vskip 0.5cm
\begin{proposition}[\bf{Correlation formula}]\label{lem:secderi}For any ${\rm BMS}(\underline{\epsilon})$ channel and any fixed linear code we have
\begin{align}\nonumber
\frac{\partial^2}{\partial\epsilon_i\partial\epsilon_j}H_n(\underline{X}\mid \underline{Y})=&\delta_{ij}\int_{-1}^{+1}
 dt_i\frac{\partial^2c_D(t_1)}{\partial\epsilon_i^2} g_1(t_i)  
\\ \nonumber & +(1-\delta_{ij})\int_{-1}^{+1}\int_{-1}^{+1}
dt_idt_j\frac{\partial c_D(t_i)}{\partial\epsilon_i}\frac{\partial c_D(t_j)}{\partial \epsilon_j}g_2(t_i,t_j)
\end{align}
with 
\begin{align}
g_2(t_i,t_j)=\expt_{\underline{t}^{\sim ij}} \biggl[
\ln \biggl(\frac{1-t_iT_i-t_jT_j+t_it_jT_{ij}}
{1-t_iT_i-t_jT_j+t_it_jT_iT_j}
\biggr)\biggr]
\nonumber 
\end{align}
\end{proposition}
Again, for the case of interest later on, we have a ${\rm BMS}(\epsilon)$ channel and a linear code ensemble that is symmetric under permutations of bits, thus
\begin{align}\label{second-derivative-ensemble}
\frac{d^2}{d\epsilon^2}\mathbb{E}_{\cal C} [h_n]=&\int_{-1}^{+1}
dt_1\frac{\partial^2c_D(t_1)}{\partial\epsilon^2} \mathbb{E}_{\cal C}[g_1(t_1)] 
\\ \nonumber & 
+\sum_{i\neq 1}\int_{-1}^{+1}\int_{-1}^{+1}
dt_1dt_i\frac{\partial c_D(t_1)}{\partial\epsilon}\frac{\partial  c_D(t_i)}{\partial\epsilon}\mathbb{E}_{\cal C}[g_2(t_1,t_i)]
\end{align}
For the BEC\footnote{The same remark than before applies here.} these formulas simplify
\begin{align}\nonumber
\frac{\partial^2}{\partial\ln \epsilon_i\partial\ln\epsilon_j}H_n(\underline{X}\mid \underline{Y})=(1-\delta_{ij})\ln 2\expt_{\underline{t}}\big[T_{ij}-T_iT_j]
\end{align}
and 
\begin{align}\label{bec-second-derivative}
\frac{d^2}{(d\ln \epsilon)^2}\mathbb{E}_{\cal C} [h_n&]=\ln 2\sum_{i\neq 1}^{n}\expt_{\code,\underline{t}}\big[T_{1i}-T_1T_i]
\end{align}
For the BIAWGNC
\begin{align}
\frac{\partial^2}{\partial\epsilon_i^{-2}\partial\epsilon_j^{-2}}H_n(\underline{X}\mid \underline{Y})=\frac{1}{2}\expt_{\underline{t}}[\big(T_{ij}-T_iT_j\big)^2], \end{align}
and
\begin{align}\label{biawgnc-second-derivative}
\frac{d^2}{(d\epsilon^{-2})^2}\mathbb{E}_{\cal C} [h_n]=\frac{1}{2}\sum_{i=1}^{n}&\expt_{\code,\underline{t}}[\big(T_{1i}-T_1T_i\big)^2]
\end{align}
Formulas \eqref{bec-second-derivative} and \eqref{biawgnc-second-derivative} involve the ``correlation'' $(T_{ij}-T_iT_j)$ for bits $X_i$ and $X_j$.
The general formula \eqref{second-derivative-ensemble} can also be recast in terms of powers of such correlations by expanding the logarithm (see section \ref{section-3}). Loosely speaking, in the infinite block length limit $n\to +\infty$, the second derivative will be well defined only if the correlations have sufficient decay with respect to the graph distance (the minimal length among all paths joining $i$ and $j$ on the Tanner graph). Thus we expect good decay properties for all noise levels except at the phase transition  thresholds where, in the limit $n\to +\infty$, the first derivative generally has bounded discontinuities, and thus the second derivative cannot be uniformly bounded in $n$.

\subsection{Relation to mutual information}\label{subsection-mutual}

  The correlation $T_{ij}-T_iT_j$ is basicaly a measure of the
independence of two codebits, thus it is natural to expect that it should be related to the mutual
information $I(X_i;X_j\mid \underline{Y})$. We do not pursue this issue in all details because it is not used in the rest of the paper, but wish to briefly state the main relations which follow naturaly form the previous formulas.

\vskip 0.5cm

\noindent{\bf The BEC$(\underline{\epsilon})$.} Take $i\neq j$. The chain rule implies $H_n(\underline{X}\mid \underline{Y})= H(X_i X_j\mid \underline{Y})+ H(\underline{X}^{\sim ij}\mid X_iX_j \underline{Y})$. Also $H(\underline{X}^{\sim ij}\mid X_iX_j \underline{Y})= H(\underline{X}^{\sim ij}\mid X_iX_j \underline{Y}^{\sim ij})$. Since $H(\underline{X}^{\sim ij}\mid X_iX_j \underline{Y}^{\sim ij})$ does not depend on $\epsilon_i, \epsilon_j$ we have 
$$
\frac{\partial^2}{\partial\epsilon_i\partial\epsilon_j}
H_n(\underline{X}\mid \underline{Y})
= 
\frac{\partial^2}{\partial\epsilon_i\partial\epsilon_j}
H(X_iX_j\mid \underline{Y})
$$
The conditional entropy on the r.h.s is explicitly
$\epsilon_i\epsilon_j H(X_iX_j|\underline{Y}^{\sim ij})+\epsilon_i(1-\epsilon_j) H(X_i|X_j\underline{Y}^{\sim ij}) +(1-\epsilon_i)\epsilon_j H(X_j|X_i\underline{Y}^{\sim ij})$. In this expression the three conditional entropies are independent of the channel parameters $\epsilon_i$ and $\epsilon_j$. Thus
\begin{align}\nonumber
\frac{\partial^2}{\partial\epsilon_i\partial\epsilon_j}
H_n(\underline{X}\mid \underline{Y})
&=
H(X_iX_j\mid\underline{Y}^{\sim ij}) - H(X_i\mid X_j\underline{Y}^{\sim ij})
- H(X_j\mid X_i\underline{Y}^{\sim ij})
\\ \nonumber &
= H(X_j\mid\underline{Y}^{\sim ij}) - H(X_j\mid X_i\underline{Y}^{\sim ij})
\\ \nonumber &
= I(X_i;X_j\mid \underline{Y}^{\sim ij})
=\frac{1}{\epsilon_i\epsilon_j} I(X_i;X_j\mid\underline{Y})
\nonumber
\end{align}
Summarizing, we have obtained for $i\neq j$,
\begin{align}\nonumber
\frac{\partial^2}{\partial\ln\epsilon_i\partial\ln\epsilon_j}
H_n(\underline{X}\mid \underline{Y})
=I(X_i;X_j\mid\underline{Y})
= \mathbb{E}_{\underline t}[T_{ij}-T_iT_j]
\end{align}

\vskip 0.5cm

\noindent{\bf The BIAWGNC$(\underline{\epsilon})$.} Take $i\neq j$. We note that 
$$
T_{ij}= p_{X_iX_j\mid\underline Y}(00\mid\underline y) + p_{X_iX_j\mid\underline Y}(11\mid\underline y) - p_{X_iX_j\mid\underline Y}(01\mid\underline y) - p_{X_iX_j\mid\underline Y}(10\mid\underline y)
$$
from which it follows
$$
(T_{ij}-T_iT_j)^2
\leq 4\sum_{x_i,x_j}
\biggl\vert p_{X_iX_j\mid\underline Y}(x_ix_j\mid\underline y) - p_{X_i\mid \underline Y}(x_i\mid \underline y)p_{X_j\mid \underline Y}(x_j\mid \underline y)\biggr\vert ^2
$$
Applying the inequality 
$$
\frac{1}{2}\sum_x\vert P(x) - Q(x)\vert^2\leq  D(P\Vert Q)
$$
for the Kullback-Leibler divergence of the two distributions $P=p_{X_iX_j\mid\underline y}$ and $Q=p_{X_i\mid \underline Y}p_{X_j\mid \underline y}$,
we get for $i\neq j$
$$
(T_{ij}-T_iT_j)^2\leq 8 I(X_i;X_j\mid \underline y)
$$
Averaging over the outputs we get
$$
\frac{\partial^2}{\partial\epsilon_i^{-2}\partial\epsilon_j^{-2}}H_n(\underline X\mid \underline Y)
= \mathbb{E}_{\underline{t}}[(T_{ij}-T_iT_j)^2]
\leq 8 I(X_i;X_j\mid \underline{Y})
$$

\vskip 0.5cm

\noindent{\bf Highly noisy BMS channels.} From the high noise expansion
(see section \ref{section-3} and the above remarks, we can derive an inequality like the preceding one, which holds in the high noise regime for general BMS channels. The number $8$ gets replaced by some suitable factor which depends on the channel noise.

\subsection{Organisation of the paper}

The statistical mechanics formulation is very convenient to perform many of the necessary calculations, but also the interpolation method is best formulated in that framework. Thus we briefly recall it in section \ref{section-2} as well as a few connections to the information theoretic language. Section \ref{section-3} contains the derivation of the correlation formula (proposition \ref{lem:secderi}) and other useful material. The interpolation method that is used to prove the variational bound (theorem \ref{thm:main}) is presented in section \ref{section-4}. The main new ingredient of the proof is an estimate (see proposition \ref{conj:overlaps} in section \ref{section-4}) on the fluctuations of overlap parameters. The proof of proposition \ref{conj:overlaps} is the object of section \ref{section-5}. The appendices contain technical calculations involved in the proofs.

\section{Statistical Mechanics Formulation}\label{section-2}
Consider a fixed code belonging to the ensemble $\code= LDPC(n,\Lambda,P)$. The posterior distribution $p_{\underline{X}\vert \underline{Y}}(\underline{x}\vert \underline{y})$ used in MAP decoding can be viewed 
as the Gibbs measure of a particular random spin system. For this it is convenient to use the usual mapping of bits onto spins $\sigma_i=(-1)^{x_i}$. Given any set $A\subset\{1,...,n\}$, we use the notation $\sigma_A=\prod_{i\in } \sigma_i$. Thus $\sigma_A=(-1)^{\oplus_{i\in A} x_i}$. It will be clear from the context if the subscript is a set or a single bit. 
For a uniform prior over the code words and a BMS channel, Bayes rule implies $p_{\underline{X}\vert \underline{Y}}(\underline{x}\vert \underline{y})= \mu(\underline{\sigma})$ with
$$
\mu(\underline{\sigma})=\frac{1}{Z}\prod_{c}\frac{1}{2}(1+\sigma_{\partial c})\prod_{i=1}^n e^{\frac{l_i}{2}\sigma_i}
$$
where $\prod_{c}$ is a product over all check nodes of the given code, and $\sigma_{\partial c}=\prod_{i\in \partial c}\sigma_i$ is the product of the spins (mod $2$ sum of the bits) attached to the variable nodes $i$ that are connected to a check $c$. $Z$ is the normalization factor or ``partition function'' and $\ln Z$ is the ``pressure'' associated to the Gibbs measure $\mu(\underline{\sigma})$. It is related to the conditional entropy by 
\begin{align}
H_n(\underline{X}\vert\underline{Y})
= \mathbb{E}_{\underline{l}}[\ln Z]-\sum_{i=1}^{n}\int_{-\infty}^{+\infty} dl_i c_L(l_i) \frac{l_i}{2}\label{eq:condentropy}
\end{align}
Expectations with respect to $\mu(\underline\sigma)$ for a fixed graph and a fixed channel output 
are denoted by the bracket $\langle-\rangle$. More precisely for any 
$A\subset\{1,...,n\}$,
$$
\langle \sigma_A\rangle=\sum_{\sigma^n} \sigma_A \mu(\sigma^n),\qquad \sigma_A=\prod_{i\in A} \sigma_i
$$ 
More details on the above formalism can be found for example in \cite{MacrisIT07}.

The soft estimate of the bit $X_i$ is (in the difference domain)
\begin{align}\label{eq:intrinsic}
T_i=\langle\sigma_i\rangle
\end{align}
We will also need soft estimates
for $X_i\oplus X_j$, $i\neq j$. In the statistical mechanics formalism they are simply
expressed as 
\begin{align}\label{eq:Tij}
T_{ij}=\langle\sigma_i\sigma_j\rangle  
\end{align}
In particular the correlation between bits $X_i$ and $X_j$ becomes $T_{ij}-T_i T_j=\langle\sigma_i\sigma_j\rangle - \langle\sigma_i\rangle \langle\sigma_j\rangle$, which is the usual notion of spin-spin correlation in statistical mechanics. 

In section \ref{section-3} (and appendices \ref{appendix-B}, 
\ref{appendix-C}) the algebraic manipulations 
are best performed in terms of ``extrinsic'' soft bit estimates. We will need many variants, the simplest one being the estimate of $X_i$ when observation $y_i$ is not available
\begin{align*}
T_i^{\sim i}=\tanh \frac{L_i^{\sim i}}{2}=p_{X_i\vert \underline{Y}^{\sim i}}(0\vert \underline{y}^{\sim i})- p_{X_i\vert \underline{Y}^{\sim i}}(1\vert \underline{y}^{\sim i})
\end{align*}
The second is the estimate of  $X_i$ when both $y_i$ and $y_j$ are not available
\begin{align*}
T_{i}^{\sim ij}=\tanh\frac{L_{i}^{\sim ij}}{2}=p_{X_i\vert \underline{Y}^{\sim ij}}(0\vert \underline{y}^{\sim ij})- p_{X_i\vert \underline{Y}^{\sim ij}}(1\vert \underline{y}^{\sim ij}) 
\end{align*}
Finally we will also need the extrinsic estimate of the mod $2$ sum $X_i\oplus X_j$ when both $y_i$ and $y_j$ are not available,
\begin{align*}
T_{ij}^{\sim ij}=\tanh\frac{L_{ij}^{\sim ij}}{2}=p_{X_i\oplus X_j\vert \underline{Y}^{\sim ij}}(0\vert \underline{y}^{\sim ij})- p_{X_i\oplus X_j\vert \underline{Y}^{\sim ij}}(1\vert \underline{y}^{\sim ij}) 
\end{align*}
It is practical to work in terms of a modified Gibbs average $\langle \sigma_A \rangle_{\sim i}$ which means that $l_i=0$ , in other words $y_i$ is not available. Similarly we introduce the averages $\langle \sigma_X \rangle_{\sim ij}$, in other words both $y_i$ and $y_j$ are unavailable. One has
\begin{equation}\nonumber
T_i^{\sim i}=\langle\sigma_i\rangle_{\sim i}, \qquad T_i^{\sim ij}=\langle\sigma_i \rangle_{\sim ij}, \qquad T_{ij}^{\sim ij}=\langle \sigma_i \sigma_j\rangle_{\sim ij}
\end{equation}
The extrinsic brackets $\langle-\rangle_{\sim i}$ and $\langle-\rangle_{\sim ij}$ are related to the usual ones
$\langle - \rangle$ by the following formulas derived in appendix \ref{appendix-A},
\begin{equation}\label{sigmatosigmao}
\langle\sigma_i\rangle_{\sim i}=\frac{\langle\sigma_i\rangle -t_i}{1-\langle\sigma_i\rangle t_i}
\end{equation}
and 
\begin{equation}\label{sigmai}
\langle\sigma_i\rangle_{\sim ij}={\displaystyle \frac{\langle\sigma_i\rangle- t_i-\langle\sigma_i\sigma_j\rangle t_j+ t_i t_j\langle\sigma_j\rangle}{1-\langle\sigma_i\rangle t_i-\langle\sigma_j\rangle t_j+\langle\sigma_i\sigma_j\rangle t_i t_j}}
\end{equation}
\begin{align}\label{sigmatosigmaoo}
\langle\sigma_i\sigma_j\rangle_{\sim ij}=\frac{\langle\sigma_i\sigma_j\rangle-t_i\langle\sigma_j\rangle -\langle\sigma_i\rangle t_j+ t_i t_j}{1-\langle\sigma_i\rangle t_i-\langle\sigma_j\rangle t_j+\langle\sigma_i\sigma_j\rangle t_i t_j} 
\end{align}

\section{The Correlation Formula}\label{section-3}

A derivation of propositon \ref{lem:firstderi} and of \eqref{eq:becfderi}, \eqref{eq:awgnfderi} within the formalism outlined in section \ref{section-2} can be found in \cite{Macris}.

\subsection{Proof of proposition \ref{lem:secderi}}
For any ${\rm BMS}(\underline\epsilon)$ channel and linear code  we have from \eqref{eq:condentropy}
\begin{align*}
\frac{\partial}{\partial\epsilon_i}H_n(\underline{X}\mid\underline{Y})
&
=\mathbb{E}_{\underline{l}^{\sim j}}\biggl[\int_{-\infty}^{+\infty} dl_j \frac{\partial c_L(l_j)}{\partial\epsilon_j} (\ln Z-\frac{l_j}{2})\biggr] 
\end{align*}
The second equality follows by permutation symmetry of code bits.
Differentiating once more, we get
\begin{align}
\frac{\partial^2}{\partial\epsilon_i\partial\epsilon_j}H_n(\underline{X}\mid\underline{Y})=\delta_{ij}S_1+(1-\delta_{ij})S_2 \label{eq:secderi1}
\end{align}
where 		
\begin{align}
S_1=\expt_{\underline{l}^{\sim i }}\biggl[\int_{-\infty}^{+\infty}
dl_i\frac{\partial^2c_L(l_i)}{\partial\epsilon_i^2}(\ln Z-\frac{l_i}{2}) \biggr]\label{eq:I}
\end{align}
and
\begin{align}\nonumber  
S_2=\expt_{\underline{l}^{\sim ij}}\biggl[\int_{-\infty}^{+\infty}
dl_idl_j\frac{\partial c_L(l_i)}{\partial \epsilon_i}\frac{\partial c_L(l_j)}{\partial\epsilon_j}(\ln Z-\frac{l_i}{2})\biggr]
\end{align}

First we consider $S_1$. Let 
$$
Z_{\sim i}=\sum_{\underline{\sigma}}\prod_{c\in\code}\frac{1}{2}(1+\sigma_{\partial c})\prod_{k\neq i} e^{\frac{l_k}{2}\sigma_k}
$$ 
be the partition function for the Gibbs measure $\langle -\rangle_{\sim i}$ introduced in section \ref{section-2} and
consider
$$
\ln\frac{Z}{Z_{\sim i}}=\ln\langle e^{\frac{l_i}{2}\sigma_i}\rangle_{\sim i}
$$
Using the identity 
\begin{equation}\label{iden}
e^{\frac{l_i}{2}\sigma_i}= e^{\frac{l_i}{2}}\frac{1+t_i\sigma_i}{1+t_i}
\end{equation}
we get
$$
\ln Z - \frac{l_i}{2}=\ln Z_{\sim i}+\ln\biggl(\frac{1+t_i\langle\sigma_i\rangle_{\sim i}}{1+ t_i}\biggr)
$$
When we replace this expression in the integral \eqref{eq:I} we see that the contribution of $\ln Z_{\sim i}$ vanishes because this later quantity is independent of $l_i$. Indeed
$$
\int_{-\infty}^{+\infty}
dl_i\frac{\partial ^2c_L(l_i)}{\partial \epsilon_i^2}\ln Z_{\sim i}=
\ln Z_{\sim i} \,\frac{\partial ^2}{\partial \epsilon_i^2}\int_{-\infty}^{+\infty}
dl_1 c_L(l_i) =0
$$
since $c_L(l_i)$ is a normalized probability distribution.
Then, using \eqref{sigmatosigmao} leads to
\begin{align}
S_1&=  \int_{-1}^{+1}
dt_i\frac{\partial^2c_D(t_i)}{\partial\epsilon_i^2} \expt_{ \underline{t}^{\sim i}}\biggl[\ln\biggl(\frac{1+t_i\langle\sigma_i\rangle_{\sim i}}{1+ t_i}\biggr)
\biggr]
\label{eqIextrinsic} \\ &
=- \int_{-1}^{+1}
dt_i\frac{\partial^2c_D(t_i)}{\partial\epsilon_i^2} \expt_{ \underline{t}^{\sim i}}\biggl[\ln\biggl(\frac{1-t_i\langle\sigma_i\rangle}{1 - t_i}\biggr)\biggr]
\label{eqIintrinsic}
\end{align}
which (because of \eqref{eq:intrinsic}) coincides with the first term in the correlation formula.

Now we consider the term $S_2$. Notice that
\begin{align}\nonumber
\int_{-\infty}^{+\infty}
dl_idl_j\frac{\partial c_L(l_i)}{\partial\epsilon_i}\frac{\partial c_L(l_j)}{\partial \epsilon_j}\frac{l_j}{2}
= \int_{-\infty}^{+\infty}
dl_j\frac{\partial c_L(l_j)}{\partial\epsilon_j}\frac{l_j}{2}
\frac{\partial}{\partial\epsilon_i}\int_{-\infty}^{+\infty}dl_i c_L(l_i)=0  
\end{align} 
Thus we can rewrite $S_2$ as
\begin{align}\nonumber
S_2=\expt_{\underline{l}^{\sim ij}}\biggr[\int_{-\infty}^{+\infty}
dl_idl_j\frac{\partial c_L(l_i)}{\partial\epsilon_i}\frac{\partial c_L(l_j)}{\partial\epsilon_j}(\ln Z-\frac{l_i}{2}-\frac{l_j}{2})\biggl]
\end{align}
Let
$Z_{\sim ij}=\sum_{\underline{\sigma}}\prod_{c\in{\cal C}}\frac{1}{2}(1+\sigma_{\partial c})\prod_{k\neq i,j} e^{\frac{l_k}{2}\sigma_k}$ be the partition function for the Gibbs measure $\langle\cdot\rangle_{\sim ij}$, and
consider 
\begin{align}\nonumber
\ln\frac{Z}{Z_{\sim ij}}&=\ln\langle e^{\frac{l_i}{2}\sigma_i+\frac{l_j}{2}\sigma_j}\rangle_{\sim ij} 
\end{align}
Using again \eqref{iden} we get
\begin{align}\nonumber
\ln Z-\frac{l_i}{2}-\frac{l_j}{2}=\ln Z_{\sim ij} + \ln \biggl(\frac{1+t_i\langle\sigma_i\rangle_{\sim ij}+t_j\langle\sigma_j\rangle_{\sim ij}+t_it_j\langle\sigma_i\sigma_j\rangle_{\sim ij}}{1+t_i+t_j+t_it_j}\biggr) 
\end{align} 
As before the contribution of $\ln Z_{\sim ij}$ vanishes because it is  independent of $l_i$, $l_j$. Similarly we have
\begin{align}\nonumber
\int_{-\infty}^{+\infty}
dl_idl_j\frac{\partial c_L(l_i)}{\partial\epsilon_i}\frac{\partial c_L(l_j)}{\partial\epsilon_j} \ln (1+t_i\langle\sigma_i\rangle_{\sim ij})&=
{\rm same~with~i~and~j~exchanged}
\\ \nonumber &
=0 
\end{align}
\begin{align}\nonumber
\int_{-\infty}^{+\infty}
dl_idl_j\frac{\partial c(l_i)}{\partial\epsilon_i}\frac{\partial c(l_j)}{\partial\epsilon_j} \ln (1+t_i)&=
{\rm same~with~i~and~j~exchanged}
\\ \nonumber &
=0 
\end{align}
Using these four identities then leads to
\begin{align}\label{eq:IIextrinsic}
S_2= &\expt_{\underline{l}^{\sim ij}}\biggl[\int_{-1}^{+1}
dt_idt_j\frac{\partial c_D(t_i)}{\partial\epsilon_i}\frac{\partial c_D(t_j)}{\partial\epsilon_j}
\\ \nonumber &
\times \ln \bigg(\frac{1+t_i\langle\sigma_i\rangle_{\sim ij}+t_j\langle\sigma_j\rangle_{\sim ij}+t_it_j\langle\sigma_i\sigma_j\rangle_{\sim ij}}{1+t_i\langle\sigma_i\rangle_{\sim ij}+t_j\langle\sigma_j\rangle_{\sim ij}+t_it_j\langle\sigma_i\rangle_{\sim ij}\langle\sigma_j\rangle_{\sim ij}}\bigg)\biggr]
\end{align}
To get the formulas in terms of usual averages we use the relations 
(\ref{sigmai}), (\ref{sigmatosigmaoo}).
Hence
\begin{align}
S_2= \expt_{\underline{l}^{\sim ij}}\biggl[\int_{-1}^{+1}
dt_idt_j\frac{\partial c_D(t_i)}{\partial\epsilon_i}\frac{\partial c_D(t_j)}{\partial\epsilon_j}\ln \bigg(\frac{1-t_i\langle\sigma_i\rangle -t_j\langle\sigma_j\rangle +t_it_j\langle\sigma_i\sigma_j\rangle}{1-t_i\langle\sigma_i\rangle - t_j\langle\sigma_j\rangle +t_it_j\langle\sigma_i\rangle\langle\sigma_j\rangle}\bigg)\biggr]
\label{S2final}
\end{align}
Because of \eqref{eq:intrinsic} and \eqref{eq:Tij} this coincides with the second term in the correlation formula. The proposition now follows from \eqref{eq:secderi1}, \eqref{eqIintrinsic} and \eqref{S2final}.  

\subsection{Expressions in terms of the spin-spin correlation}

\vskip 0.5cm

\noindent{\bf The BEC.} 
From $c_D(t)= (1-\epsilon)\delta(t-1) +\epsilon\delta(t)$, the second derivative in terms of extrinsic quantities (formulas \eqref{eqIextrinsic} and \eqref{eq:IIextrinsic}) reduces to
\begin{align*}
\frac{\partial^2}{\partial\epsilon_i\partial\epsilon_j}H_n(\underline X\mid \underline Y)&=(1-\delta_{ij})\expt_{\underline{t}^{\sim ij}}\biggl[\ln\biggl(\frac{1+\langle\sigma_i\rangle_{\sim ij}+\langle\sigma_j\rangle_{\sim ij}+\langle\sigma_i\sigma_j\rangle_{\sim ij}}
{1+\langle\sigma_i\rangle_{\sim ij}+\langle\sigma_j\rangle_{\sim ij}+\langle\sigma_i\rangle_{\sim ij}\langle\sigma_j\rangle_{\sim ij}}\biggr)\biggr]
\end{align*}
There are various ways to see that for the BEC any Gibbs average $\langle \sigma_A\rangle$ or $\langle\sigma_A\rangle_{\sim ij}$ takes values in $\{0,1\}$. A heuristic explanation is that bits (or their mod $2$ sums) are either perfectly known or erased. A more formal explanation follows from a Nishimori identity\footnote{We will use various such identities. A  proof of their most general form can be found in \cite{MacrisIT07}.  A general reference is \cite{Nishimori-book}.}
combined with the Griffith-Kelly-Sherman (GKS) correlation inequality \cite{MacrisIT07}. For example, $\mathbb{E}[\langle \sigma_A\rangle^2]=
\mathbb{E}[\langle \sigma_A\rangle]$ (Nishimori) and $\langle \sigma_A\rangle\geq 0$ (GKS). Thus $\langle\sigma_A\rangle(1-\langle\sigma_A\rangle)$ is a positive random variable with zero expectation and is therefore equal to $0$ with probability one. These remarks imply that
\begin{align*}
\frac{\partial^2}{\partial\epsilon_i\partial\epsilon_j}H_n(\underline X\mid \underline Y)&=\frac{1}{\epsilon_i\epsilon_j}(1-\delta_{ij})\expt_{\underline{t}}\biggl[\ln\biggl(\frac{1+\langle\sigma_i\rangle+\langle\sigma_j\rangle+\langle\sigma_i\sigma_j\rangle}
{1+\langle\sigma_i\rangle+\langle\sigma_j\rangle+\langle\sigma_i\rangle\langle\sigma_j\rangle}\biggr)\biggr]
\end{align*}
Note that in deriving the last expression we used the fact that $l_i=\infty$ ($l_j=\infty$) implies that $\sigma_i=+1$ ($\sigma_j=+1$) which makes the logarithm term equal to zero. From the previous remarks we also have
\begin{align*}
\ln\bigl(1+\langle\sigma_i\rangle+\langle\sigma_j\rangle+\langle\sigma_i\sigma_j\rangle\bigr)= &
(\ln 2)\bigl(\langle\sigma_i\rangle+\langle\sigma_j\rangle+\langle\sigma_i\sigma_j\rangle\bigr) \nonumber \\ &+(\ln
3-2\ln
2)\bigl(\langle\sigma_i\rangle\langle\sigma_j\rangle+\langle\sigma_i\rangle\langle\sigma_i\sigma_j\rangle+\langle\sigma_j\rangle\langle\sigma_i\sigma_j\rangle\bigr) \nonumber \\
&+(5\ln 2-3\ln 3)\langle\sigma_i\rangle
\langle\sigma_j\rangle\langle\sigma_i\sigma_j\rangle
\end{align*} 
and
\begin{align*}
\ln\bigl(1+\langle\sigma_i\rangle+\langle\sigma_j\rangle+\langle\sigma_i\rangle\langle\sigma_j\rangle\bigr)=(\ln 2)\bigl(\langle\sigma_i\rangle+\langle\sigma_j\rangle\bigr) 
\end{align*}
The difference of the two logarithms is simplified using the following four Nishimori identities,
\begin{align*}
\expt_{\underline{t}}[\langle\sigma_i\rangle\langle\sigma_j\rangle]
=
\expt_{\underline{t}}[\langle\sigma_i\rangle\langle\sigma_i\sigma_j\rangle] 
=
\expt_{\underline{t}}[\langle\sigma_j\rangle\langle\sigma_i\sigma_j\rangle]
=\expt_{\underline{t}}[\langle\sigma_i\sigma_j\rangle\langle\sigma_i\rangle\langle\sigma_j\rangle]
\end{align*}
Finaly we obtain the simple expression
\begin{align*}
\frac{\partial^2}{\partial\epsilon_i\partial\epsilon_j} H_n(\underline X\mid \underline Y)&=\frac{\ln 2}{\epsilon_i\epsilon_j}(1-\delta_{ij})\expt_{\underline{t}}\bigl[\langle\sigma_i\sigma_j\rangle-\langle\sigma_i\rangle\langle\sigma_j\rangle\bigr]
\nonumber \\
&
=
\frac{\ln 2}{\epsilon_i\epsilon_j}(1-\delta_{ij})\expt_{\underline{t}}\bigl[T_{ij}-T_i T_j\bigr]
\end{align*}
Let us point out that the second GKS inequality (for the BEC) implies that $\langle\sigma_i\sigma_j\rangle-\langle\sigma_i\rangle\langle\sigma_j\rangle\geq 0$, thus the correlation takes values in $\{0,1\}$ and we have 
$\expt_{\underline{t}}\bigl[T_{ij}-T_i T_j\bigr]=\expt_{\underline{t}}\bigl[(T_{ij}-T_i T_j)^2\bigr]$.

\vskip 0.5cm

\noindent{\bf The BIAWGNC.} From the explicit form  
\begin{equation}\nonumber
 c_L(l)= \frac{1}{\sqrt{2\pi\epsilon^{-2}}}\exp\biggl(-\frac{(l-\epsilon^{-2})^2}{2\epsilon^{-2}}\biggr)
\end{equation}
 one can show that the correlation formula reduces to  
\begin{align*}
\frac{\partial^2}{\partial\epsilon_i^{-2}\partial\epsilon_j^{-2}}H_n(\underline X\mid\underline Y)&= \mathbb{E}_{\underline{t}}[(\langle \sigma_i \sigma_j\rangle - \langle \sigma_i\rangle\langle \sigma_j\rangle)^2]
\nonumber \\ &
=
\mathbb{E}_{\underline{t}}\bigl[\big(T_{ij}-T_iT_j\big)^2\bigr]
\nonumber
\end{align*}
Otherwise diffrentiating \eqref{eq:condentropy} thanks to
\begin{equation}\nonumber
 \frac{d^2 c_L(l)}{(d\epsilon^{-2})^2}= \biggl(-\frac{\partial}{\partial l}+\frac{\partial^2}{\partial l^2}\biggr)^2 c_L(l)
\end{equation}
and using integration by parts also leads to this simpler form. This route is much simpler and the details can be found in \cite{MacrisIT07}.

\vskip 0.5cm
\noindent{\bf Highly noisy BMS channels.} We use the extrinsic form of 
the correlation formula given by \eqref{eqIextrinsic} and \eqref{eq:IIextrinsic}.
First we expand the logarithms in $S_1$ and $S_2$ in powers of $t_i$ and $t_j$ and then use various Nishimori identities. After some tedious algebra (see Appendices \ref{appendix-B} and \ref{appendix-C}) we can organize the expansion in powers of the channel parameters \eqref{parameters}.
In the high noise regime this expansion is absolutely convergent.
To lowest order we have
\begin{align}
&\frac{\partial^2}{\partial\epsilon_i\partial\epsilon_j}H_n(
\underline{X}\mid\underline{Y})
= \delta_{ij}S_1+(1-\delta_{ij})S_2
 \nonumber
 \\ &
 \approx \frac{1}{2} \delta_{ij}{m_2}^{(2)} (\expt_{\underline{l}}[\langle\sigma_i\rangle^{2}]-1) + 
\frac{1}{2}(1-\delta_{ij}) [{m_1}^{(2)}]^2\expt_{ \underline{l}}\biggl[\bigl(\langle\sigma_i\sigma_j\rangle
 - \langle\sigma_i\rangle\langle\sigma_j\rangle\bigr)^2\biggr] + \ldots
 \nonumber
 \\ &
 = \frac{1}{2} \delta_{ij}{m_2}^{(2)} (\expt_{\underline{t}}[T_i^{2}]-1) + 
\frac{1}{2}(1-\delta_{ij}) [{m_1}^{(2)}]^2\expt_{ \underline{t}}\biggl[\bigl(T_{ij} - T_i T_j\bigr)^2\biggr] + \ldots
\end{align}
The second derivative of the conditional entropy is directly related to the average square of the code-bit or spin-spin correlation. 

\section{The Interpolation Method}\label{section-4}

We use the interpolation method in the form  developed by Montanari. 
As explained in \cite{Montanari} it is difficult to establish directly the bounds for the standard ensembles. 
 Rather, one introduces a ``multi-Poisson'' ensemble which approximates the standard ensemble. 
Once the bounds are derived for the multi-Poisson ensemble they are extended to the standard ensemble by a limiting procedure. The interpolation construction is fairly complicated so that it helpful to briefly review the simpler pure Poisson case. 

\subsection{Poisson ensemble}

We introduce the ensemble Poisson-LDPC$(n,1-r,P)= {\cal P}$ where $n$ is the block length, $r$ the rate and
$P(x)=\sum_{k}P_k x^k$ the check degree distribution. 
A bipartite graph from the Poisson ensemble is constructed
as follows. The graph has $n$ variable nodes. For any $k$ choose a
Poisson number $m_k$ of check nodes with mean $n(1-r)P_k$. Thus graph
has a total of $m=\sum_k m_k$ check nodes which is also a Poisson variable with mean $n(1-r)$. For each check node $c$ of degree
$k$, choose $k$ variable nodes uniformly at random and connect them to $c$. One
can show that the left degree distribution concentrates around a Poisson distribution $\Lambda_{\cal P}(x)=e^{P^\prime(1)(1-r)(x-1)}$. In other words the fraction $\Lambda_l$ of variable nodes with degree $l$ is Poisson with mean $P^\prime(1)(1-r)$.

The main idea
behind the interpolation technique is to recursively remove the check node constraints and compensate them
 with extra observations $U$ distributed as \eqref{U-variable} where $d_V$ is a trial distribution to be optimized in the final inequality. One can interpret these extra observations as coming from a repetition code whose rate is tuned in a such a way that the total design rate $r$ remains fixed. More precisely
let $s\in[0,1]$ be 
an interpolating parameter.
At ``time'' $s$ we have a Poisson-LDPC$(n, (1-r)s,P)={\cal P}_s$ code. Besides the usual channel outputs $l_i$, each node $i$ receives $e_i$ extra i.i.d observations $U_a^i$, $a=1,...,e_i$, where $e_i$ is Poisson  with mean $n(1-r)(1-s)$ (so the total effective rate is fixed to $r$). The interpolating Gibbs measure is
\begin{align}
\mu_{s}(\underline\sigma)&= \frac{1}{Z_{s}}\prod_{c}\frac{1}{2}(1+\sigma_{\partial c})\prod_{i=1}^n e^{(\frac{l_i}{2}+\sum_{a=1}^{e_i}U_a^{i})\sigma_i} \label{eq:interpolated_measure}
\end{align}
Here $\prod_c$ is a product over checks of a given graph in the ensemble ${\cal P}_s$.
At $s=1$ one recovers the original measure while at $s=0$ (no checks) we have a simple product measure (corresponding to a repetition code) which is tailored to yield the 
replica symmetric entropy $h_{RS}[d_V;\Lambda_{\cal P},P]$ (up to an extra constant).

The central result of \cite{Montanari} is the sum rule
\begin{align}
\expt_{{\cal P}}[h_n]=h_{RS}[d_V;\Lambda_{\cal P}, P]+\int^{1}_{0}R_n(s) ds \label{eq:calculus} 
\end{align}
Let us explain the notation. The first term on the right hand side 
$h_{RS,{\cal P}}[d_V;\Lambda_{\cal P}, P]$ is the replica symmetric functional of section \ref{section-1} evaluated for the Poisson ensemble.
The remainder term $R_n(s)$ is
\begin{align*}
R_n(s)=\sum_{p=1}^{\infty} \frac{1}{2p(2p-1)} \expt_s\biggl[\big\langle P(Q_{2p})-P'(q_{2p})(Q_{2p}-q_{2p})-P(q_{2p})\big\rangle_{2p,s}\biggr]
\end{align*}
with $q_{2p}=\expt_{V}[(\tanh V)^{2p}]$ and 
$Q_{2p}$ the overlap parameters  
\begin{equation}\label{overlap-poisson}
Q_{2p}=\frac{1}{n}\sum_{i=1}^{n}\sigma_i^{(1)}\sigma_i^{(2)}\cdots\sigma_i^{(2p)}
\end{equation}
Here $\sigma_{i}^{(\alpha)}, \alpha=1,2,\dots,2p$ are $2p$ independent copies (replicas) of the spin $\sigma_i$ and  
$\langle-\rangle_{2p,s}$ is the Gibbs bracket associated to the product measure (replica measure)
$$
\prod_{\alpha=1}^{2p} \mu_{s}(\underline\sigma^{(\alpha)})
$$

\subsection{Multi-Poisson ensemble}
The multi-Poisson-LDPC$(n,\Lambda,P,\gamma)= {\cal MP}$ ensemble, is a more elaborate construction which allows to approximate a target 
${\rm LDPC}(n,\Lambda,P)$ ensemble. Its parameters are the block length $n$, the target  variable and check node degree distributions $\Lambda(x)$ and $P(x)$ and the real number $\gamma$ which 
controls the closeness to the standard ensemble. We recall that variable and check node degrees have finite maximum degrees.
The construction of a bipartite graph from the multi-Poisson ensemble proceeds via rounds: the process
starts with a high rate code and at each round one adds a very small number of check
nodes till one ends up with a code

with almost the desired rate and degree distribution. A graph process $\graph_t$ is defined for discrete times  $t=0,..., t_{max}$, $t_{max}=\lfloor\Lambda'(1)/\gamma\rfloor-1$ as
follows. For $t=0$, $\graph_0$ has no check nodes and has $n$ variable nodes.
The set of variable nodes is partitioned
into the subsets $\mathcal{V}_l$ of cardinality $n\Lambda_l$ for every $l$ and every node $i\in \mathcal{V}_l$ is decorated with $l$ free sockets. The number $d_i(t)$ keeps track of the number of free sockets on node $i$ once round $t$ is completed. So for $t=0$, $\graph_0$ has no check nodes and each variable node $i\in \mathcal{V}_l$ has $d_i(0)=l$ free sockets.
At round $t$, $\graph_{t}$ is constructed from $\graph_{t-1}$ as follows. For all $k$, choose a Poisson number $m_k^t$ of check nodes with mean $n\gamma P_k/P^\prime(1)$. Connect each outgoing edge of these new degree $k$ check nodes
(added at time $t$) to variable node $i$ according to the probability
$w_i(t)=\frac{d_i(t-1)}{\sum_i d_i(t-1)}$. This is the fraction of free sockets at node $i$ after round $t-1$ was completed.
Once all new check nodes are
connected, update the number of free sockets for each variable node
$d_i(t)=d_i(t-1)-\Delta_i(t)$.
where $\Delta_i(t)$ is the number of times the
variable node $i$ was chosen during the round $t$. 
For $n\to \infty$ this construction yields graphs with variable degree distributions $\Lambda_\gamma(x)$ (the check degree distribution remains $P(x)$). The variational distance between $\Lambda_\gamma(x)$ and $P(x)$ tends to zero as $\gamma\to 0$. 

The interpolating ensemble now uses two parameters $(t_*,s)$ where $t_*\in \{0,...,t_{max}\}$ and $0\leq s\leq \gamma$. For rounds $0,...,t_*-1$ one proceeds exactly as before to obtain a graph
 $\graph_{t_*-1}$.  At the next round $t_*$, one proceeds as before but with $\gamma$ replaced by $s$.
The rate loss is
compensated by adding $e_i$ extra observations for each node $i$, where $e_i$ is a Poisson integer with mean
$n(\gamma-s)w_i(t_*)$. The round is ended by updating the number of free sockets $d_i(t_*)=d_i(t_*-1) - \Delta_i(t_*) - e_i(t_*)$.
Finally, for rounds after $t_*+1,...,t_{max}$ no new check node
is added but for each variable node $i$, $e_i$ external observations are added, where $e_i$ is a Poisson integer with mean
$n\gamma w_i(t_*)$. Moreover the free socket counter is updated as 
$d_i(t)=d_i(t-1)-e_i(t)$. Recall that the external observations are i.i.d copies of the random variable $U$ (see \eqref{U-variable}).

The interpolating Gibbs measure $\mu_{t_*,s}(\underline\sigma)$ has the same form than (\ref{eq:interpolated_measure}) with the appropriate products over checks and extra observations. Let $h_{n,\gamma}$ the conditional entropy of the multi-Poisson ensemble ${\cal MP}$ (corresponding to $t_*=t_{max}$ and $s=\gamma$). Again, the central result of \cite{Montanari} is the sum rule
\begin{equation}\label{eq:interMP}
\expt_{{\cal MP}}[h_{n,\gamma}] = h_{RS}[d_V; \Lambda_\gamma, P]+{\displaystyle
\sum_{t_*=0}^{t_{max}-1}\int^{\gamma}_{0}R_n(t_*,s) ds}+o_n(1)
\end{equation}
Explanations on the notation are in order. The first term $h_{RS, \gamma}[d_V;\Lambda_\gamma,P]$ is the replica symmetric functional of \ref{section-1} evaluated for the multi-Poisson ensemble.
The remainder term $R_n(t_*,s)$ is given by 
\begin{align}
R_n(t_*,s)=\sum_{p=1}^{\infty} \frac{1}{(2p)(2p-1)}\expt_s\biggl[\bigl\langle P(Q_{2p})-P'(q_{2p})(Q_{2p}-q_{2p})-P(q_{2p})\bigr\rangle_{2p,t_*,s}\biggr] \label{eq:mpinter}
\end{align}
where $q_{2p}=\expt_{V}[(\tanh V)^{2p}]$ as before and  
$Q_{2p}$ are modified overlap parameters
\begin{equation}\label{modified-overlap}
Q_{2p}={\displaystyle \sum_{i=1}^{n}w_i(t_*)X_i(t_*)\sigma_i^{(1)}\sigma_i^{(2)}\cdots\sigma_i^{(2p)} } 
\end{equation}
Here as before $\sigma_{i}^{(\alpha)}, \alpha=1,2,\dots,2p$ are $2p$ independent copies (replicas) of the spin $\sigma_i$ and  
$\langle-\rangle_{2p,t_*,s}$ is the Gibbs bracket associated to the product measure
$$
\prod_{\alpha=1}^{2p} \mu_{{t_*,s}}(\underline\sigma^{(\alpha)})
$$ 
The overlap parameter is now more complicated than in the Poisson case because of the (positive) terms $w_i(t_*)$ and $X_i(t_*)$. Here $X_i(t_*)$ are new i.i.d random variables whose precise description is quite technical and can be found in \cite{Montanari}. The reader may think of the terms $w_i(t_*)X_i(t_*)$ as behaving like the $\frac{1}{n}$ factor of
the pure Poisson ensemble overlap parameter \eqref{overlap-poisson}. More precisely the only properties (see Appendix E in \cite{Montanari}) that we need are
\begin{equation}\label{w-prop}
 \sum_{i=1}^n w_i(t_*)=1,\qquad\mathbb{P}\bigl[w_i(t_*)\leq \frac{A}{n}\bigr]\geq 1- e^{-Bn}
\end{equation}
and 
\begin{equation}\label{X}
0\leq X_i(t_*)\leq x,\qquad
\mathbb{E}[x^k]\leq A_k
\end{equation}
for any finite $k$ and finite positive constants $A$, $B$, $A_k$ independent of $n$ (they may depend on some of the other parameters but this turns out to be unimportant).
Finaly we use the shorthand 
$\expt_s[-]$ for 
the expectation with respect to all random variables involved in the interpolation measure. The subscript $s$ is here to remind us that this expectation depends on $s$, afact that is important to keep in mind because the remainder involves an intgral over $s$. When we use $\mathbb{E}$ (without the subscript $s$; as in \eqref{X} for example) it means that the quantity does not depend on $s$. In the sequel the replcated Gibbs bracket $\langle -\rangle_{2p,t_*,s}$ is simply denoted by $\langle - \rangle_s$. There will be no risk of confusion because the only property that we us is its linearity. 

In \cite{Montanari} it is shown that
\begin{equation}\label{monta}
\expt_{\code}[h_n]=\expt_{{\cal MP}}[h_{n,\gamma}]+ O(\gamma^b) +o_n(1)
\end{equation}
where
$O(\gamma^b)$ is uniform in $n$ ($b>0$ a numerical constant) and $o_n(1)$ (depends on $\gamma$) tends to $0$ as $n\to +\infty$. 

In the next paragraph we prove the variational bound
 on the conditional entropy of the multi-Poisson ensemble, namely
\begin{equation}\label{varia}
 \liminf_{n\to +\infty}\expt_{{\cal MP}}[h_{n,\gamma}] \geq 
 h_{RS}[d_V;\Lambda_\gamma,P]
\end{equation}
Note that here $o_n(1)$ again depends on $\gamma$. 
 By combining this bound with \eqref{monta} and taking limits
\begin{equation}
 \liminf_{n\to+\infty}\expt_{\code}[h_n] = \lim_{\gamma\to 0}\liminf_{n\to+\infty}\expt_{{\cal MP}}[h_{n,\gamma}] \geq 
\lim_{\gamma\to 0} h_{RS}[d_V;\Lambda_\gamma, P] = h_{RS}[d_V;\Lambda, P]
\end{equation}
The main theorem \ref{thm:main} then follows by maximizing the right hand side over $d_V$.

\subsection{Proof of the Variational Bound \eqref{varia}}

In view of the sum rule \eqref{eq:interMP} it is sufficient to prove that $\liminf_{n\to +\infty}R_n(t_*,s)\geq 0$.  In the case of a convex $P$ considered in \cite{Montanari} this is immediate because convexity is equivalent to 
\begin{equation*}
 P(Q_{2p}) -P(q_{2p})\geq P^\prime(q_{2p})(Q_{2p}-q_{2p})
\end{equation*}
Note that $P(x)=\sum_kP_kx^k$ is anyway convex for $x\geq 0$ since all $P_k\geq 0$. So if do not assume convexity of the check node degree distribution we have to circumvent the fact that $Q_{2p}$ can be negative. But note
\begin{align*}
\langle Q_{2p}\rangle&={\displaystyle \sum_{i=1}^{n}w_i(t_*) X_i(t_*)\langle
\sigma_i^{(1)}\sigma_i^{(2)}\cdots\sigma_i^{(2p)}\rangle } \nonumber \\
&={\displaystyle \sum_{i=1}^{n}w_i(t_*) X_i(t_*)\langle
\sigma_i^{(1)}\rangle\langle\sigma_i^{(2)}\rangle\cdots\langle\sigma_i^{(2p)}\rangle
} \nonumber \\
&={\displaystyle \sum_{i=1}^{n}w_i(t_*)X_i(t_*)\langle
\sigma_i\rangle^{2p}}\ge 0 
\end{align*}
Therefore we are assured that for any $P$ (i.e not necessarily convex for $x\in\mathbb{R}$) we have 
\begin{equation}\label{conv}
 P(\langle Q_{2p}\rangle) -P(q_{2p})\geq P^\prime(q_{2p})(\langle Q_{2p}\rangle-q_{2p})
\end{equation}
and the proof will follow if we can show that with high probability
\begin{equation*}
 P(Q_{2p}) \approx P(\langle Q_{2p}\rangle)
\end{equation*}
The following concentration estimate will suffice and is proven in section \ref{section-5}. 

\begin{proposition}\label{conj:overlaps}
Fix any $\delta<\frac{1}{4}$. 
On the BEC($\epsilon$) and BIAWGNC($\epsilon$) for a.e $\epsilon$, or on general BMS($\epsilon$) satisfying $H$, we have for a.e $\epsilon$,
\begin{equation}
\lim_{n\to \infty}\int_{0}^{\gamma}ds \prob_s\biggl[|P(Q_{2p})-P(\langle
Q_{2p}\rangle_s)|>\frac{2p}{n^{\delta}}\biggr]=0 
\end{equation} 
Here
$\prob_s(X)$ is the probability distribution $\expt_s\langle \mathbb{I}_X\rangle_{s}$.
\end{proposition}

This proposition can presumably be strengthened in two directions. First we conjecture that hypothesis $H$ is not needed (this is indeed the case for the BEC and BIAWGNC). Secondly the statement should hold for all $\epsilon$
 except at a finite set of threshold values of $\epsilon$ where the conditional entropy is not  differentiable, and its first derivative is expected to have jumps (except for cycle codes where higher order derivatives are singular). Since we are unable to control the locations of theses jumps our proof only works for Lebesgue almost every $\epsilon$.

We are now ready to complete the proof of the variational bound \eqref{varia}.

\vskip 0.5cm

\noindent{\it End of Proof of \eqref{varia}}. From \eqref{modified-overlap} and \eqref{X}
\begin{equation}\label{overbound}
 \vert Q_{2p} \vert \leq \sum_{i=1}^n w_i(t_*) X_i(t_*) \leq x
\end{equation}
and 
\begin{align}\label{power-k}
\mathbb{E}_s[\langle Q_{2p}^k\rangle_s] \leq A_k
\end{align}
Combined with $q_{2p}\leq 1$, this implies (since the maximal degree of $P$ is finite) that
\begin{equation}\label{remainder-bound}
\expt_s[\langle P(Q_{2p}) -P^\prime(q_{2p})(Q_{2p}-q_{2p}) - P(q_{2p})\rangle_s]\leq C_1
\end{equation}
for some positive constant $C_1$. The only crucial feature here is that this constant does not depend on $n$ and on the number of replicas $2p$ (a more detailed analysis shows that it depends only on the degree of $P(x)$).

Now we split the sum \eqref{eq:mpinter} into terms with $1\leq p\leq n^\delta$ (call this contribution $R_A$) and terms with $p\geq n^\delta$ (call this contribution $R_B$), where $\delta>0$ is the constant of proposition \ref{conj:overlaps}.
For the second contribution \eqref{remainder-bound} implies
\begin{equation}
 R_B\leq C_1\sum_{p\geq n^{\delta}}\frac{1}{2p(2p-1)} =  O(n^{-\delta})
\end{equation}
For the first contribution we write
\begin{align}
 R_A= &\sum_{p\leq n^\delta} \frac{1}{2p(2p-1)} \expt_s[\langle P(Q_{2p})\rangle_s-
P(\langle Q_{2p}\rangle_s)]
\nonumber\\&
+
\sum_{p\leq n^\delta} \frac{1}{2p(2p-1)} 
\expt[P(\langle Q_{2p})\rangle_s) - P^\prime(q_{2p})(\langle Q_{2p})\rangle_s - q_{2p}) -P(q_{2p})]
\nonumber
\end{align}
In this equation, the second sum is positive due to \eqref{conv}. Thus we find
\begin{align}\nonumber
 R_n(t_*,s)&=R_A+R_B
\\ \nonumber &
\geq \sum_{p\leq n^\delta} \frac{1}{2p(2p-1)} \expt[\langle P(Q_{2p})\rangle_s
-
P(\langle Q_{2p}\rangle_s)]
-  O(n^{-\delta})
\nonumber
\end{align}
Below we use proposition \ref{conj:overlaps} to show that for almost every $\epsilon$ in the appropriate range
\begin{equation}\label{fluctuation}
 \lim_{n\to +\infty}\int_0^\gamma ds \sum_{p\leq n^\delta} \frac{1}{2p(2p-1)} \expt_s[\langle P(Q_{2p})\rangle_s -
P(\langle Q_{2p}\rangle_s)]=0
\end{equation}
which implies by Fatou's lemma
\begin{equation}\nonumber
 \liminf_{n\to +\infty}\sum_{t_*=0}^{t_{max}-1}\int_0^\gamma R_n(t_*, s) ds \geq 0
\end{equation}
and thus proves \eqref{varia} for almost every $\epsilon$ in the appropriate range. A general convexity argument allows to extend this result to all $\epsilon$ in the same range. Indeed convexity arguments imply that both sides of the inequality \eqref{varia} are continuous functions of $\epsilon$\footnote{At this point one could use arguments involving physical degradation if BMS($\epsilon$) is degraded as a function of $\epsilon$. But we take a more direct route that does not assume physical degraddation as a function of $\epsilon$}. To show continuity of the left hand side we use inequality \eqref{C} in Appendix \ref{appendix-C}: it implies that there exists a positive number $\rho$ (independent of $\epsilon$ and $n$) such that $\frac{d^2}{d\epsilon^2}\mathbb{E}_s[h_{n,\gamma}]\geq - \rho$. Therefore  $\mathbb{E}_s[h_{n,\gamma}] +\frac{\rho}{2}\epsilon^2$ is convex in $\epsilon$; so the $\liminf_{n\to +\infty}$ is also convex and thus continuous on any open $\epsilon$ set. To show continuity of the right hand side we first note that for each $d_V$, $h_{RS}$ is a linear functional of the channel distribution $c_L(l)$; thus the $\sup_{d_V}$ is a convex functional of 
$c_L(l)$; thus it is continuous in any open $\epsilon$ where
$c_L(l)$ varies smoothly in $\epsilon$ (this last point can be made more precise using tools from functional analysis).

Let us now prove \eqref{fluctuation}. First we set
\begin{equation*}
 F_{2p}=\bigl\vert \langle P(Q_{2p})\rangle_s -
P(\langle Q_{2p}\rangle_s)\bigr\vert
\end{equation*}
 and use Cauchy-Schwarz and then \eqref{power-k} to obtain
\begin{align}
 \expt[F_{2p}] & = \expt_s\bigl[F_{2p}\mathbb{I}_{F_{2p} \leq \frac{2p}{n^\delta}}\bigr]
+ \expt_s\bigl[F_{2p}\mathbb{I}_{F_{2p}\geq \frac{2p}{n^\delta}}\bigr]
\nonumber \\ &
\leq 
\frac{2p}{n^\delta} + \expt_s\bigl[ F_{2p}^2\bigr]^{1/2} \mathbb{P}_s\bigl[F_{2p}\geq \frac{2p}{n^\delta}\bigr]^{1/2}
\nonumber \\ &
\leq
\frac{2p}{n^\delta} + C_2 \mathbb{P}_s\bigl[F_{2p}\geq \frac{2p}{n^\delta}\bigr]^{1/2}
\nonumber 
\end{align}
for some positive constant $C_2$ independent of $n$ and $p$ (depending only on the degree of $P(x)$). 
Thus 
\begin{align}
 \int_0^\gamma ds & \sum_{p\leq n^\delta}\frac{1}{2p(2p-1)}\expt[F_{2p}]
\nonumber\\&
\leq 
\frac{1}{n^\delta}\sum_{p\leq n^\delta}\frac{1}{2p-1}
+
C_2\int_0^\gamma ds \sum_{p\leq n^\delta}\frac{1}{2p(2p-1)}
\mathbb{P}_s\bigl[F_{2p}\geq \frac{2p}{n^\delta}\bigr]^{1/2}
\nonumber \\ &
\leq 
O\bigl(\frac{\ln n^\delta}{n^{\delta}}\bigr)
+
C_2\sum_{p\leq n^\delta}\frac{\sqrt{\gamma}}{2p(2p-1)}\biggl(\int_0^\gamma ds
\mathbb{P}_s\bigl[F_{2p}\geq \frac{2p}{n^\delta}\bigr]\biggr)^{1/2}
\nonumber
\end{align}
In the second inequality we have permuted the integral with a finite sum and  used Cauchy-Schwarz. Finaly we can apply proposition \ref{conj:overlaps} and Lebesgue's dominated convergence theorem to the last sum over $p$, to conclude that \eqref{fluctuation} holds.

\section{Fluctuations of overlap parameters}\label{section-5} 
In this section we prove proposition \ref{conj:overlaps}.
The proofs are done directly for the multi-Poisson ensemble. 
We start by a relation between the overlap fluctuation and the spin-spin
correlation.

\begin{lemma}
For any BMS($\epsilon$) channel there exists a finite constant $C_3$
independent of $n$ and $p$ (dpending only on the maximal check degree) such that
\begin{equation}
 \mathbb{P}_s\biggl[\bigl\vert P(Q_{2p})- P(\langle Q_{2p}\rangle_s)\bigr\vert\geq \frac{2p}{n^\delta}\biggr]
\leq 
\frac{C_3}{p^2n^{2\delta-\frac{1}{2}}}
\biggl(\sum_{i=1}^n\mathbb{E}_s\bigl[(\langle \sigma_1\sigma_i\rangle_s - \langle\sigma_1\rangle_s
\langle\sigma_i\rangle_s)^2\bigr]\biggr)^{1/2}
\end{equation}
\end{lemma}

\begin{proof}
 Using the identity 
\begin{equation}
Q_{2p}^k-\langle Q_{2p}\rangle_{s}^k = (Q_{2p}-\langle
Q_{2p}\rangle_{s})\sum_{l=0}^{k-1}Q_{2p}^{k-l-1}\langle
Q_{2p}\rangle_{s}^l
\end{equation}
and \eqref{X} we get
\begin{align}
\vert P(Q_{2p}) - \langle P(Q_{2p}\rangle_s\vert & = 
\vert Q_{2p}-\langle Q_{2p}\rangle_s\vert\vert\sum_{k} P_k\sum_{l=0}^{k-1}Q_{2p}^{k-l-1}\langle Q_{2p}\rangle_s^l\vert
\nonumber\\&
\leq 
\vert Q_{2p}-\langle Q_{2p}\rangle_s\vert\sum_k k P_k x^{k-1}
\nonumber\\&
\leq
P^\prime(x) \vert Q_{2p}-\langle Q_{2p}\rangle\vert
\nonumber
\end{align}
Here $x$ is the bound in \eqref{overbound}. 
Therefore applying the Chebycheff inequality
\begin{equation}\label{chebycheff}
\mathbb{P}_s\biggl[|P(Q_{2p})-P(\langle 
Q_{2p}\rangle_s)|\geq\frac{2p}{n^{\delta}}\biggr]
\leq
\frac{n^{2\delta}}{4p^2}
\expt_s\biggl[P^\prime(x)^2\bigl(\langle Q_{2p}^2\rangle_s-\langle
Q_{2p}\rangle_s^2\bigr)\biggr]
\end{equation}
From the definition of the overlap parameters it follows that
\begin{align}
 \langle Q_{2p}^2\rangle_s - \langle Q_{2p}\rangle_s^2
& =
\sum_{i,j=1}^{n}w_i(t_*)w_j(t_*)X_i(t_*)X_j(t_*)\bigl(\langle\sigma_i\sigma_j
\rangle_s^{2p} - \langle\sigma_i\rangle_s^{2p}\langle\sigma_j\rangle_s^{2p}\bigr)
\nonumber\\&
\leq 
2p \sum_{i,j=1}^n x^2w_i(t_*) w_j(t_*)\bigl(\langle\sigma_i\sigma_j\rangle - \langle\sigma_i\rangle\langle\sigma_j\rangle\bigr)
\nonumber
\end{align}
Substituting in \eqref{chebycheff} and applying Cauchy-Schwarz to $\sum_{i,j}\expt_s[-]$ we get
\begin{align}\nonumber
\mathbb{P}_s\bigl[|P(Q_{2p})-
P(\langle Q_{2p}\rangle_s)|\geq\frac{2p}{n^{\delta}}\bigr]
 \leq &
\frac{n^{2\delta}}{2p} 
\biggl(\sum_{i,j=1}^n \expt_s[x^4P^\prime(x)^4w_i(t_*)^2 w_j(t_*)^2  ]\biggr)^{1/2}
\nonumber\\&
\times
 \biggl(\sum_{i,j=1}^n\expt_s[(\langle\sigma_i\sigma_j\rangle_s - \langle\sigma_i\rangle_s\langle\sigma_j\rangle_s)^2]\biggr)^{1/2}
\nonumber
\end{align}
From \eqref{w-prop}, \eqref{X} it is easy to see that for any $i,j$
\begin{align}
\expt_s[x^4P^\prime(x)^4w_i(t_*)^2 w_j(t_*)^2 ]
\leq \frac{C_3^2}{n^4}
\nonumber
\end{align}
where $C_3$ is independent of $n$.
It follows that
\begin{align}\nonumber
 \mathbb{P}_s\biggl[|P(Q_{2p}) & -
P(\langle Q_{2p}\rangle_s)|\geq \frac{2p}{n^{\delta}}\biggr]
\\ \nonumber &  
 \leq
\frac{n^{2\delta-1}}{2p}C_3 
 \biggl(\sum_{i,j=1}^n\expt_s[(\langle\sigma_i\sigma_j\rangle_s - \langle\sigma_i\rangle_s\langle\sigma_j\rangle_s)^2]\biggr)^{1/2}
\nonumber \\ &
=\frac{n^{2\delta-\frac{1}{2}}}{2p}C_3
\biggl(\sum_{i=1}^n\expt_s[(\langle\sigma_i\sigma_1\rangle_s - \langle\sigma_i\rangle_s\langle\sigma_1\rangle_s)^2]\biggr)^{1/2}
\label{missing}
\end{align}
In the last equality we have used the symmetry of the ensemble with respect to variable node permutations.
\end{proof}

Denote by $h_{n,\gamma}(t_*,s)$ the entropy of the $\mu_{t_*,s}$ interpolating measure. Note that this should not be confused with the multi-Poisson ensemble entropy $h_{n,\gamma}$ (which corresponds to $t_*=t_{\max}$ and $s=\gamma$).

\begin{lemma}\label{lemma-3}
 For the BEC and BIAWGNC with any noise value and for general BMS($\epsilon$) channels satisfying $H$ we have
\begin{equation}\label{inequ}
 \sum_{i=1}^n\expt_s[(\langle\sigma_1\sigma_i\rangle_s - \langle\sigma_1\rangle_s
\langle\sigma_i\rangle_s)^2]\leq F(\epsilon)+ G(\epsilon)\frac{d^2}{d\epsilon^2}\expt_s[h_{n,\gamma}(t_*,s)]
\end{equation}
where $F(\epsilon)$ and $G(\epsilon)$ are two finite constants depending only on the channel parameter.
\end{lemma}

The proof of lemma \ref{lemma-3} is based on the correlation formula of section \ref{section-1}. These are true for any linear code ensemble so they are in particular true for the interpolating $(t_*,s)$ ensemble\footnote{in fact one has to check that the addition of $\sum_{a=1}^{e_i} U_a^i$ to $l_i$ does not change the derivation and the final formulas. For this it sufices to follow the calculation of section \ref{section-3}}.
For the BEC and BIAWGNC we have already shown the two equalities \eqref{bec-second-derivative} and \eqref{biawgnc-second-derivative}: thus the inequality \eqref{inequ} is in fact an equality for appropriate values of $F$ and $G$.
The case of general (but highly noisy) BMS channels is presented in appendix \ref{appendix-C}. A converse inequality can also be proven by the methods of appendices \ref{appendix-B} and \ref{appendix-C}.

\vskip 0.5cm

\begin{proof}[{\it Proof of proposition \ref{conj:overlaps}.}] Note that for all points of the parameter space $(\epsilon,s)$ such that the second derivative of the average conditional entropy is bounded uniformly in $n$ the proof immediately follows from \eqref{missing}, \eqref{inequ} (and the last inequality before that one) by choosing $\delta<\frac{1}{4}$. However, in the large block length limit $n\to+\infty$, genericaly the first derivative of the average conditional entropy has jumps for some threshold values of $\epsilon$ (these values depend on the interpolation parameter $s$). This means that for these threshold values the second derivative cannot be bounded uniformly in $n$. Since we cannot control these locations we introduce a test function $\psi(\epsilon)$:
non negative, infinitely differentiable and with small enough bounded support included in the range of $\epsilon$ satisfying $H$.
We consider the averaged quantity
\begin{equation}\label{Q}
\mathcal{Q}=\int d\epsilon \psi(\epsilon)\int_0^\gamma ds 
\mathbb{P}_s\biggl[|P(Q_{2p})-
P(\langle Q_{2p}\rangle)|\geq\frac{2p}{n^{\delta}}\biggr]
\end{equation}
Writing $\psi(\epsilon)=\sqrt{\psi(\epsilon)}\sqrt{\psi(\epsilon)}$  Cauchy-Schwarz implies
\begin{equation*}
\mathcal{Q}\leq \int_0^\gamma ds\biggl(\int d\epsilon \psi(\epsilon) \mathbb{P}_s\biggl[|P(Q_{2p})-
P(\langle Q_{2p}\rangle)|\geq\frac{2p}{n^{\delta}}\biggr]^2\biggr)^{1/2}
\end{equation*}
Combining this inequality with \eqref{missing} and \eqref{inequ} we get
\begin{align}
\mathcal{Q} & \leq \frac{ n^{2\delta-\frac{1}{2}}}{2p} C_3\int_0^\gamma ds\biggl(\int d\epsilon \psi(\epsilon) \bigl(F(\epsilon)+G(\epsilon)\frac{d^2}{d\epsilon^2}\mathbb{E}_s[h_{n,\gamma}(t_*,s)]\bigr)\biggr)^{1/2}
\nonumber \\ &
=
\frac{ n^{2\delta-\frac{1}{2}}}{2p} C_3\int_0^\gamma ds\biggl( \int d\epsilon \psi(\epsilon)F(\epsilon)-\int d\epsilon  \frac{d}{d\epsilon}\bigl(\psi(\epsilon)G(\epsilon)\bigr)\frac{d}{d\epsilon}\mathbb{E}_s[h_{n,\gamma}(t_*,s)]\biggr)^{1/2}
\nonumber
\end{align}
Note that from the bounds in appendix \ref{appendix-C} $F(\epsilon)$, $G(\epsilon)$ and $G^\prime(\epsilon)$ are integrable except possibly at the edge of the $\epsilon$ range defined by $H$. This is not a problem because we can take the support of $\psi(\epsilon)$ away from such points or alternatively take a $\psi(\epsilon)$ which vanishes 
sufficiently fast at these points.
Moreover the first derivative of the average conditional entropy is bounded uniformly in $n$ and $s$ (see appendix \ref{appendix-E}) by a constant $k(\epsilon)$ that has at most a power singularity at $\epsilon=0$, and again this is not a problem. Thus  by choosing $0<\delta< \frac{1}{4}$ we obtain
\begin{equation*}
 \lim_{n\to+\infty} \mathcal{Q} = 0
\end{equation*}
Applying Lebesgue's dominated convergence theorem to convergent subsequences (of the integrand of $\int d\epsilon \psi(\epsilon)$ in \eqref{Q}) we deduce that
\begin{equation*}
 \int d\epsilon\psi(\epsilon)\lim_{n_k\to +\infty}\int_0^\gamma ds 
\mathbb{P}_s[|P(Q_{2p})-
P(\langle Q_{2p}\rangle_s)|\geq\frac{2p}{n_k^{\delta}}] = 0
\end{equation*}
which implies that along any convergent subsequences, for almost all $\epsilon$
\begin{equation}
 \lim_{n_k\to +\infty}\int_0^\gamma ds 
\mathbb{P}_s\biggl[|P(Q_{2p})-
P(\langle Q_{2p}\rangle_s)|\geq\frac{2p}{n_k^{\delta}}\biggr] = 0
\end{equation}
as long as $\delta\leq \frac{1}{4}$. Now we apply this last statement to  two subsequences that attain the $\liminf$ and the $\limsup$ (on the intersection of the two measure one $\epsilon$ sets). This proves that the $\lim_{n\to +\infty}$ exists and vanishes.
\end{proof}

\section{Conclusion}\label{conclusion}

The main new tool introduced in this paper are relationships between the second derivative of the conditional entropy and correlation functions or mutual information between code bits. This allowed us to estimate the overlap fluctuations in order to get a better handle on the remainder.
Some aspects of our analysis bear some similarity with techniques introduced by Talagrand \cite{Talagrand-book} but is independent. One difference is that we use specific symmetry properties of the communications problem.

We expect that the technique developped here can be extended 
to remove the restriction to high noise (condition $H$). Indeed the only place in the analysis where we need this restriction is lemma \ref{lemma-3}. For the BEC and BIAWGNC the lemma is trivialy satisfied for any noise level (with appropriate constants). Another issue that would be worthwhile investigating is whether the related inequalities of paragraph \ref{subsection-mutual} and the converse of lemma \ref{lemma-3} can be derived irrespective of the noise level.

The next obvious problem is to prove the converse of the variational bound (theorem \ref{thm:main}).

For this one should show that the remainder vanishes when $d_V$ is replaced by the maximizing distribution of $h_{RS}[d_V; \Lambda, P]$. This program has been carried out explicitely in the case of the BEC and the Poisson ensemble \cite{Kudekar-Korada-Macris-nice}. It would be desirable to extend this to more general ensembles and channels but the problem becomes quite hard. However a similar program has been succesfuly carried out for a p-spin model with gauge symmetry\footnote{In the present context gauge symmetry and channel symmetry are equivalent} (see \cite{Korada-Macris-seattle}). A solution of these problems would allow for a rigorous determination of MAP thresholds and would extend our understanding of the intimate relationship between BP and MAP decoding.

\appendix

\section{Appendix A}\label{appendix-A}
We prove the identities \eqref{sigmatosigmao},
\eqref{sigmai}
, \eqref{sigmatosigmaoo}. By definition
\begin{equation*}
 \langle\sigma_i e^{-\frac{l_i}{2}\sigma_i}\rangle
=\frac{1}{Z}\sum_{\underline\sigma}\sigma_i\prod_c\frac{1}{2}(1+\sigma_{\partial c})\prod_{j\neq i}e^{\frac{l_i}{2}\sigma_i}
\end{equation*}
and 
\begin{equation*}
 \langle e^{-\frac{l_i}{2}\sigma_i}\rangle
=\frac{1}{Z}\sum_{\underline\sigma}\prod_c\frac{1}{2}(1+\sigma_{\partial c})\prod_{j\neq i}e^{\frac{l_i}{2}\sigma_i}
\end{equation*}
Thus
\begin{equation*}
 \langle\sigma_i\rangle_{\sim i}= \frac{\langle\sigma_i e^{-\frac{l_i}{2}\sigma_i}\rangle}{\langle e^{-\frac{l_i}{2}\sigma_i}\rangle}
\end{equation*}
and plugging the identity
\begin{align*}
e^{-\frac{l_i}{2}\sigma_i}=e^{-\frac {l_i}{2}}\frac{1-\sigma_it_i}{1-t_i}
\end{align*}
 in the brackets immediately leads to \eqref{sigmatosigmao}.
For the second and third identities we proceed similarly. Namely,
\begin{equation*}
 \langle\sigma_i\rangle_{\sim ij}= \frac{\langle\sigma_i e^{-\frac{l_i}{2}\sigma_i}e^{-\frac{l_j}{2}\sigma_j}\rangle}{\langle e^{-\frac{l_i}{2}\sigma_i}e^{-\frac{l_j}{2}\sigma_j}\rangle}
\end{equation*}
and 
\begin{equation*}
 \langle\sigma_i\sigma_j\rangle_{\sim ij}= \frac{\langle\sigma_i\sigma_j e^{-\frac{l_i}{2}\sigma_i}e^{-\frac{l_j}{2}\sigma_j}\rangle}{\langle e^{-\frac{l_i}{2}\sigma_i}e^{-\frac{l_j}{2}\sigma_j}\rangle}
\end{equation*}
Plugging 
\begin{align*}
e^{-\frac{l_i}{2}\sigma_i}e^{-\frac{l_j}{2}\sigma_j}=e^{\frac{l_i+l_j}{2}}\frac{1-\sigma_it_i-\sigma_j t_j +\sigma_i\sigma_j t_i t_j}{1-t_i-t_j + t_it_j}
\end{align*}
in the brackets, leads immediately to \eqref{sigmai} and \eqref{sigmatosigmaoo}.

\section{Appendix B}\label{appendix-B}

We indicate the main steps of the derivation of the full high noise expansion for 
$$
\frac{\partial^2}{\partial\epsilon_i\partial\epsilon_j}H_n(\underline X\mid\underline Y) =\delta_{ij} S_1 +(1-\delta_{ij}) S_2
$$
The expansion for $S_1$ is given by \eqref{S1-exp} and that for $S_2$ by \eqref{S2-exp}.
They are derived in a form that is suitable to prove lemma \ref{lemma-3} of section \ref{section-5} (see appendix \ref{appendix-C}).
For this later proof we need to extract a square correlation at each order as in \eqref{S2-exp}. This is achieved here through the use of appropriate remarquable Nishimori identities, and in order to use these we take the extrinsic forms \eqref{eqIextrinsic} and \eqref{eq:IIextrinsic} of $S_1$ and $S_2$.

Let us start with $S_1$ which is simple. Using the power series expansion of $\ln(1+x)$ we have
\begin{equation*}
 \ln\biggl(\frac{1+t_i\langle\sigma_i\rangle_{\sim i}}{1+t_i}\biggr)
= \sum_{p=1}^{+\infty} \frac{(-1)^{p+1}}{p}t_1^{p}(\langle\sigma_i\rangle_{\sim i}^p-1)
\end{equation*}
This yields an infinite series for $S_1$ which we will now simplify.
Because of the Nishimori identities 
\begin{equation*}
 \mathbb{E}[t_i^{2p-1}] = \mathbb{E}[t_i^{2p}],\qquad
 \mathbb{E}_{\underline t^{\sim i}}[\langle\sigma_i\rangle_{\sim i}^{2p-1}] = \mathbb{E}_{\underline t^{\sim i}}[\langle\sigma_{\sim i}\rangle_1^{2p}]
\end{equation*}
we can combine odd and even terms and get
\begin{equation}\label{S1-exp}
 S_1 = \sum_{p=1}^{+\infty}\frac{m_2^{(2p)}}{2p(2p-1)}\bigl(\mathbb{E}_{\underline {t}^{\sim i}}[\langle\sigma_i\rangle_{\sim i}^{2p}] -1\bigr)
\end{equation}
This series is absolutely convergent as long as 
\begin{equation*}
 \sum_{p=1}^{+\infty}\frac{m_2^{(2p)}}{2p(2p-1)}< +\infty
\end{equation*}
which is true for channels satisfying $H$.

In the rest of the appendix we deal with $S_2$ which is considerably more complicated. However the general idea is the same as above. First we use the expansion of $\ln(1+x)$ to get
\begin{equation}\label{log-expansion}
 \ln\Bigg(\frac{1+\langle\sigma_i\rangle_{\sim ij }t_i+\langle\sigma_j\rangle_{\sim ij}t_j+\langle\sigma_i\sigma_j\rangle_{\sim ij}t_it_j}{1+\langle\sigma_i\rangle_{\sim ij}t_i+\langle\sigma_j\rangle_{\sim ij}t_j+\langle\sigma_i\rangle_{\sim ij}\langle\sigma_j\rangle_{\sim ij}t_it_j}\Bigg) = {\rm I} - {\rm II} - {\rm III}
\end{equation}
where
\begin{align*}
{\rm I}=\sum_{p=1}^{\infty}\frac{(-1)^{p+1}}{p}\bigg(\langle\sigma_i\rangle_{\sim ij}t_i+\langle\sigma_j\rangle_{\sim ij}t_j+\langle\sigma_i\sigma_j\rangle_{\sim ij}t_it_j\bigg)^p
\end{align*}
\begin{align*}{\rm II}=\sum_{p=1}^{\infty}\frac{(-1)^{p+1}}{p}t_i^p\langle\sigma_i\rangle_{\sim ij}^p,\qquad
{\rm III}=\sum_{p=1}^{\infty}\frac{(-1)^{p+1}}{p}t_j^p\langle\sigma_j\rangle_{\sim ij}^p 
\end{align*}
We expand the multinomial in ${\rm I}$ 
\begin{align*}
\sum_{\substack{k_a, k_b, k_c \\ k_a+k_b+k_c=p}}\frac{p!}{k_a!k_b!k_c!}t_i^{k_a+k_c} t_j^{k_b+k_c}
\langle\sigma_i\rangle_{\sim ij}^{k_a}\langle\sigma_j\rangle_{\sim ij}^{k_b}\langle\sigma_i\sigma_j\rangle_{\sim ij}^{k_c}
\end{align*}
and subtract the terms ${\rm II}$ and ${\rm III}$. Then only terms that have powers of the form $t_i^kt_j^l$ with $k, l \ge 1$ will survive in \eqref{log-expansion}. Moreover because of the identities $\mathbb{E}[t_i^{2k-1}] = \mathbb{E}[t_i^{2k}]$ and 
$\mathbb{E}[t_j^{2l-1}] = \mathbb{E}[t_j^{2l}]$
we find for $S_2$
\begin{align}
 S_2 = & \sum_{k\geq l\geq 1}^{+\infty} m_1^{(2k)} m_1^{(2l)}\bigl(T_{00}+T_{01}+T_{10}+T_{11}\bigr) \nonumber
 \\ &
+
\sum_{l> k\geq 1}^{+\infty} m_1^{(2k)} m_1^{(2l)}\bigl(T_{00}^\prime+T_{01}^\prime+T_{10}^\prime+T_{11}^\prime\bigr)
\label{first-exp}
\end{align}
with (we abuse notation by not indicating the $(kl)$ and $(ij)$ dependence in the $T$ and $T^\prime$ factors) 
\begin{align*}
T_{\kappa\lambda} = \sum_{p=2k-\kappa}^{2k-\kappa+2l-\lambda}&\frac{(-1)^{p+1}}{p}
\frac{p!}{(p-(2l-\lambda))!(p-(2k-\kappa))!(2k-\kappa+2l-\lambda-p)!}
\\ \nonumber &
\times\mathbb{E}_{\underline t^{\sim ij}}\biggl[\langle\sigma_i\rangle_{\sim ij}^{p-(2l-\lambda)}
\langle\sigma_j\rangle_{\sim ij}^{p-(2k-\kappa)}\langle\sigma_i\sigma_j\rangle_{\sim ij}^{2k-\kappa+2l-\lambda-p}\biggr]
\nonumber
\end{align*}
and 
\begin{equation*}
 T_{\kappa\lambda}^\prime= {\rm exchange~k, l~and~\kappa,\lambda~and~i, j} 
\end{equation*}
The next simplification step occurs by using the Nishimori identity for the expectation in the above formula
\begin{align*}
&\expt_{\underline{t}^{\sim ij}}\biggl[\langle\sigma_i\rangle_{\sim ij}^{m_1}\langle\sigma_j \rangle_{\sim ij}^{m_2}\langle\sigma_i\sigma_j \rangle_{\sim ij}^{m_3}\biggr] =
\expt_{\underline{t}^{\sim ij}}\biggl[\langle\sigma_i^{m_1}\sigma_j^{m_2}(\sigma_i\sigma_j)^{m_3}\rangle_{\sim ij}\langle\sigma_i\rangle_{\sim}^{m_1}\langle\sigma_j \rangle_{\sim ij}^{m_2}\langle\sigma_i\sigma_j \rangle_{\sim ij}^{m_3}\biggr] \nonumber
\end{align*}
and using $\sigma_i\in \{\pm 1\}$,
to ``linearize'' the terms $(\sigma_i\sigma_j)^{m_1}\sigma_i^{m_2}\sigma_j^{m_3}$. Tedious but straighforward algebra then yields
\begin{align*}
\sum_{\kappa,\lambda}T_{\kappa,\lambda}  
=
\sum_{p=2k-1}^{2k+2l-1}&\frac{(-1)^{p+1}}{p(p+1)}\frac{(p+1)!} {(p+1-2k)!(p+1-2l)!(2k+2l-p-1)!}
\\ \nonumber &
\times \mathbb{E}_{\underline{t}^{\sim ij}}\biggl[\langle\sigma_i\rangle_{\sim}^{p-2l+1}\langle\sigma_j\rangle_{\sim ij}^{p-2k+1}(\langle \sigma_i\sigma_j\rangle_{\sim ij}^{2k+2l-1-p}\biggr]
\nonumber
\end{align*}
A similar formula obtained by exchanging $k,l$ and $i,j$ holds for 
$\sum_{\kappa\lambda}T_{\kappa\lambda}^\prime$. Replacing these sums in \eqref{first-exp} yields a high noise expansion for $S_2$.

However this is not yet pratical for us because we need to extract a general square correlation factor $\bigl(\langle \sigma_i\sigma_j\rangle_{\sim ij} - \langle \sigma_i\rangle_{\sim ij}\langle\sigma_j\rangle_{\sim ij}\bigr)^2$. The fact that this is possible is a ``miracle'' that comes out of the Nishimori identities that were used. Setting
\begin{equation*}
X=\langle \sigma_i\rangle_{\sim ij}\langle \sigma_j\rangle_{\sim ij},\qquad
Y=\langle \sigma_i\sigma_j\rangle_{\sim ij}
\end{equation*}
and using the change of variables $m=p-2k+1$ the last expression becomes ($k\geq l$)
\begin{align*}
\mathbb{E}_{\underline{t}^\sim ij}\biggl[
\frac{\langle \sigma_i\rangle_{\sim ij}^{2k-2l}}{(2l)!}\sum_{m=0}^{2l}
(-1)^{m}{2l \choose m} X^m Y^{2l-m}(m+2k-2)\cdots(m+2k-(2l-1)) \biggr]
\end{align*}
One can check that this is equal to
\begin{align*}
\frac{\langle\sigma_i\rangle^{2k-2l}_{\sim ij}}{(2l)!}X^{2l-2k}\frac{\partial^{2l-2}}{\partial X^{2l-2}}\biggl(X^{2k-2}(X-Y)^{2l}\biggr)
\end{align*}
The latter can be checked by first expanding $(X-Y)^{2l}$ and then differentiating. On the other hand one can use the Leibnitz rule
\begin{align*}
\frac{\partial^{2l-2}}{\partial X^{2l-2}}\biggl(X^{2k-2}(X-Y)^{2l}\biggr)
=\sum_{r=0}^{2l-2}{2l-2\choose r}\frac{\partial^{r}}{\partial X^{r}}X^{2k-2}
\frac{\partial^{2l-2-r}}{\partial X^{2k-2-r}}(X-Y)^{2l}
\end{align*}
to find that the last expectation above is equal to
\begin{align*}
\mathbb{E}_{\underline{t}^\sim ij}\biggl[ (X-Y)^2
 \langle\sigma_i\rangle^{2k-2l}_{\sim ij}\sum_{r=0}^{2l-2}A_{rlk}X^{r}(X-Y)^{2l-r-2}\bigg]
\end{align*}
where
\begin{align*}
A_{rlk}&=\frac{1}{(2l)!}{\textstyle {2l-2\choose r}}[2l]_{r}[2k-2]_{2l-2-r},\qquad [m]_r=m(m-1)\cdots(m-r+1)
\end{align*}
We define $A_{011}=\frac12$.
We proceed similarly for the terms with $k<l$. Finaly one finds
\begin{align}
S_2= & \sum_{k\geq l\geq 1}m_1^{(2k)}m_1^{(2l)}\expt_{\underline{t}^{\sim ij}}\biggl[\biggl(\langle\sigma_i\sigma_j\rangle_{\sim ij}
 - \langle\sigma_i\rangle_{\sim ij}\langle\sigma_i\rangle_{\sim ij}\biggr)^2 
\langle\sigma_i\rangle^{2k-2l}_{\sim ij}
\nonumber
\\ & 
\times
\sum_{r=0}^{2l-2} A_{rlk}
\langle\sigma_i\rangle_{\sim ij}^r\langle\sigma_j\rangle_{\sim ij}^{r}
\biggl(\langle\sigma_i\rangle_{\sim ij}\langle\sigma_j\rangle_{\sim ij}-\langle\sigma_i\sigma_j\rangle_{\sim ij}\biggr)^{2l-2-r}\biggr]
\nonumber
\\  &
+
\sum_{l>k\geq 1} {\rm idem ~with ~k,l ~and~ i,j ~exchanged}
\label{S2-exp}
\end{align}

Let us now briefly justify that the series is absolutely convergent for channels satisfying $H$. We Note the following facts:
$A_{rlk}\leq {2l-2\choose r} 2^{2k-3}$ and $2^{2k-2}3^{2l-2}\leq (\frac{5}{2})^{2k+2l-4}$ for $k\geq l$ together with the version with $k,l$ exchanged. It easily follows that 
\begin{align}
|S_2|\leq &\frac{8}{625}\expt_{\underline{t}^{\sim ij}}\biggl[\biggl(\langle\sigma_i\sigma_j\rangle_{\sim ij}
 - \langle\sigma_i\rangle_{\sim ij}\langle\sigma_i\rangle_{\sim ij}\biggr)^2 \biggr]
\sum_{k,l\geq 1} \bigl(\frac{5}{2}\bigr)^{2k+2l} \vert m_1^{(2k)}m_1^{(2l)}\vert 
\label{conv-inequ}
\end{align}
Thus the series for $S_2$ is absolutely convergent as long as
\begin{equation*}
\sum_{p=1}^{+\infty} \bigl(\frac{5}{2}\bigr)^{2p} \vert m_1^{(2p)}\vert <+\infty
\end{equation*}
Note that we have not attempted to optimize the above estimates.

\section{Appendix C}\label{appendix-C}
We prove lemma \ref{lemma-3} for highly noisy general BMS channels. For this we use the high noise expansion derived in appendix \ref{appendix-B}. There it was derived for a general linear code ensemble, and this is also the framework of the proof below. Of course the result applies to the interpolating ensemble of lemma \ref{lemma-3}. Note that the the final constants $F(\epsilon)$ and $G(\epsilon)$ do not depend on the code ensemble but only on the channel. 

Consider equation \eqref{second-derivative-ensemble} for 
$\frac{d^2}{d\epsilon^2}\mathbb{E}_{\code, \underline{t}}[h_n]$. By the same estimates than those for $S_1$ in appendix \ref{appendix-B}, the first term on the right hand side is certainly greater than
\begin{equation*}
- \sum_{p=1}^{+\infty}\frac{|m_2^{(2p)}|}{2p(2p-1)} = -A
\end{equation*}
To get a lower bound for the second term we consider the series expansion given by that for $S_2$ in \eqref{S2-exp}. In that series we keep the first term corresponding to $k=l=1$, namely
\begin{equation*}
\frac12(m_1^{(2)})^2\sum_{j\neq 1}\expt_{\code,\underline{t}^{\sim 1j}}\biggl[\biggl(\langle\sigma_1\sigma_j\rangle_{\sim 1j}
 - \langle\sigma_1\rangle_{\sim 1j}\langle\sigma_j\rangle_{\sim 1j}\biggr)^2 \biggr] =B
\end{equation*}
and lower bound the rest of the series $(k,l)\neq (1,1)$ by using estimates of appendix 
\ref{appendix-B}. More precisely this part is lower bounded by
\begin{align*}
 -\Biggl(\frac{8}{625}&\biggl(\sum_{p=1}^{+\infty} \bigl(\frac{5}{2}\bigr)^{2p} \vert m_1^{(2p)}\vert\biggr)^2 -\frac12 (m_1^{(2)})^2 \Biggr)
\\ \nonumber \times &
\sum_{j\neq 1}\expt_{\code,\underline{t}^{\sim 1j}}\biggl[\biggl(\langle\sigma_1\sigma_j\rangle_{\sim 1j}
 - \langle\sigma_1\rangle_{\sim 1j}\langle\sigma_j\rangle_{\sim 1j}\biggr)^2 \biggr] = - C
\end{align*} 
Putting these three estimates together we get
\begin{equation}\label{C}
\frac{d^2}{d\epsilon^2}\mathbb{E}_{\code, \underline{t}}[h_n]\geq - A + B - C
\end{equation}
As long as the noise level is high enough so that (see $H$)
\begin{equation*}
\sum_{p=2}^{+\infty} \bigl(\frac{5}{2}\bigr)^{2p} \vert m_1^{(2p)}\vert
<  (\sqrt{2}-1)\bigl(\frac{5}{2}\bigr)^2\vert m_1^{(2)} \vert
\end{equation*}
the inequality \eqref{C} implies 
\begin{equation}\label{almost-lemma}
\sum_{j\neq 1}\expt_{\code,\underline{t}^{\sim 1j}}\biggl[\biggl(\langle\sigma_1\sigma_j\rangle_{\sim 1j}
 - \langle\sigma_1\rangle_{\sim 1j}\langle\sigma_j\rangle_{\sim 1j}\biggr)^2 \biggr]
 \leq 
 \tilde F(\epsilon) +\tilde G(\epsilon) \frac{d^2}{d\epsilon^2}\mathbb{E}_{\code, \underline{t}}[h_n]
\end{equation}
for two noise dependent positive finite constants  $\tilde F(\epsilon)$,
$\tilde G(\epsilon)$.

The final step of the proof consists in passing from the extrinsic average
$\langle -\rangle_{\sim 1j}$ in the correlation to the ordinary one $\langle - \rangle_{1j}$. This is achieved as follows. From the formulas \eqref{sigmai} and \eqref{sigmatosigmaoo} we deduce that 
$$
\langle\sigma_j\sigma_i\rangle-\langle\sigma_j\rangle\langle\sigma_i\rangle=
\big(\langle\sigma_j\sigma_i\rangle_{\sim ij}-\langle\sigma_j\rangle_{\sim ij}\langle\sigma_i\rangle_{\sim ij}\big) R_{ij}
$$
with
$$
R_{ij}=\frac{\big(1-\langle\sigma_i\rangle t_i-\langle\sigma_j\rangle t_j+\langle\sigma_i\sigma_j\rangle t_i t_j\big)^2}{(1- t_i^2)(1-t_j^2)}\le \frac{4}{(1-t_i^2)(1-t_j^2)} 
$$
a function that depends on all log-likelihood variables. 

Thus we have
\begin{align*}
(\langle\sigma_j\sigma_i\rangle-\langle\sigma_j\rangle\langle\sigma_i\rangle)^2 &=
\big(\langle\sigma_j\sigma_i\rangle_{\sim ij}-\langle\sigma_j\rangle_{\sim ij}\langle\sigma_i\rangle_{\sim ij}\big)^2 R^2_{ij} \\
&\le \big(\langle\sigma_j\sigma_i\rangle_{\sim ij}-\langle\sigma_j\rangle_{\sim ij}\langle\sigma_i\rangle_{\sim ij}\big)^2 \frac{16}{(1-t_i^2)^2(1-t_j^2)^2}
\end{align*}
Taking now the expectation $\expt_{\code, \underline{t}}$ we get
\begin{align*}
\expt_{\code,\underline{t}}\biggl[\biggl(\langle\sigma_j\sigma_i\rangle-\langle\sigma_j\rangle\langle\sigma_i\rangle\biggr)^2\biggr] 
&\le \expt_{\code, \underline{t}^{\sim ij}}\biggl[\biggl(\langle\sigma_j\sigma_i\rangle_{\sim ij}-\langle\sigma_j\rangle_{\sim ij}\langle\sigma_i\rangle_{\sim ij}\biggr)^2\biggr] \\ &\times \expt_{t_i,t_j}\biggl[\frac{16}{(1-t_i^2)^2(1-t_j^2)^2}\biggr]
\end{align*}
Since $t_i, t_j$ are independent we get
\begin{align}\label{schwarzy}
\expt_{t_i,t_j}\biggl[\frac{16}{(1-t_i^2)^2(1-t_j^2)^2}\biggr]&=16\biggl(\expt\biggl[\frac{1}{(1-t^2)^2}\biggr]\biggr)^2 
=16\biggl(\expt\biggl[\sum_{p\geq 0}(p+1)t^{2p}\biggr]\biggr)^2 \nonumber \\ &=16\biggl(\biggl[\sum_{p\geq 0}(p+1)m_0^{(2p)}\biggr]\biggr)^2
\end{align}
which converges for highly noisy channels satisfying $H$.
The result of the lemma follows by combining \eqref{almost-lemma} and \eqref{schwarzy}. The constants $F(\epsilon)$ and $G(\epsilon)$ are equal to $\tilde F(\epsilon)$ and $\tilde G(\epsilon)$ divided by the 
expression on the right hand side of the last inequality.

\section{Appendix D}\label{appendix-E}
We prove the boundedness and positivity of $\frac{d}{d\epsilon}\mathbb{E}_s[h_{n,\gamma}(t_*,s)]$ which is needed in the proof of lemma \ref{lemma-3}. 

\begin{lemma}\label{lemma1}
For the BEC and BIAWGNC with any noise level, and any BMS satisfying $H$, there exists a constant $k(\epsilon)$ independent of $n$, $\gamma$, $t_*$ and  $s$ such that
\begin{align}
0\le \frac{d}{d\epsilon}\mathbb{E}_s [h_{n,\gamma}(t_*,s)]\leq k(\epsilon) \label{eq:fdbnd}
\end{align}
For the  BEC we can take $k(\epsilon)=\frac{\ln 2}{\epsilon}$ and  for the BIAWGNC $k(\epsilon)=\frac{2}{\epsilon^{-3}}$. For general BMS channels satisfying $H$ the constant remains bounded as a function of $\epsilon$ (i.e. in the high noise regime).
\end{lemma}

Here we have stated the lemma for the multi-Poisson interpolating ensemble which is our specific need. However as the proof below shows it is independent of the specific code ensemble and the bound depends only on the channel.

\begin{proof}
We will use the GEXIT formula of lemma \ref{lem:firstderi}. Since the proposition applies for any linear code it also applies for the interpolating ensemble of interest here. In the case of the BEC and BIAWGNC we have (see \eqref{bec-intrin}, \eqref{biawgnc-intrin}
\begin{align*}
\frac{d}{d\epsilon}\mathbb{E}_s[h_{n,\gamma}(t_*,s)] = 
\frac{\ln 2}{\epsilon}(1-\mathbb{E}_s[\langle\sigma_1\rangle_s]
\end{align*}
and 
\begin{align*}
\frac{d}{d\epsilon}\mathbb{E}_s[h_{n,\gamma}(t_*,s)] = 
\frac{2}{\epsilon^3}(1-\mathbb{E}_s[\langle\sigma_1\rangle_s]
\end{align*}
The bounds of the lemma follow immediately since $-1\leq\sigma_1\leq 1$.

For highly noisy BMS channels we proceed by expansions. For this reason we have to use the ``extrinsic form'' of the GEXIT formula (analogous to \eqref{eqIextrinsic})
\begin{equation}\nonumber
\frac{d}{d\epsilon}\mathbb{E}_s[h_{n,\gamma}(t_*,s)] =\int_{-1}^{+1} dt_1 
\frac{\partial c_D(t_1)}{\partial \epsilon}
 \mathbb{E}_{s,\sim t_1} \biggl[\ln\biggl(\frac{1+t_1\langle \sigma_1\rangle_{s, \sim 1}}{1+t_1}\biggr)\biggr]
\end{equation}\nonumber
Expanding the logarithm and using Nishimori identities (as in the expansion of $S_1$ in appendix \ref{appendix-B}  we  obtain 
\begin{equation}\nonumber
\frac{d}{d\epsilon}\mathbb{E}_s[h_{n,\gamma}(t_*,s)] = 
\sum_{p=1}^\infty \frac{m_1^{(2p)}}{2p(2p-1)}
 \mathbb{E}_{s,\sim 1}[\langle\sigma_1\rangle_{s,\sim 1}^{2p}-1]
\end{equation}
The positivity follows from $m_1^{(2p)}\leq 0$ \cite{Ruediger-Ridchardson-book} and $-1\leq\sigma_1\leq 1$.
The upper bound (and absolute convergence) follow from condition $H$. In particular we get
\begin{equation}\nonumber
 k(\epsilon)= 2\sum_{k=1}^{+\infty}\frac{\vert m_1^{(2p)}\vert}{2p(2p-1)}
\end{equation}
which is independent of $n$, $\gamma$, $t_*$ and $s$.
\end{proof}

\section*{Acknowledgment} 
The work of Shrinivas Kudekar has been supported by a grant of the Swiss National Foundation 200021-105604. The authors acknowledge various discussions with S. Korada, O. Leveque and R. Urbanke.

\end{document}